\documentclass[11pt]{article}
\usepackage[english]{babel}
\usepackage[utf8]{inputenc}
\usepackage[a4paper, lmargin=2.5cm, rmargin=2.5cm, tmargin=3cm, bmargin=3cm]{geometry}
\usepackage{graphicx} 
\usepackage{wrapfig}
\usepackage{mathtools}
\usepackage{dsfont}
\usepackage{soul}
\usepackage{tikz}
\usepackage{authblk}
\usetikzlibrary{arrows,cd,positioning}
\usepackage{tikz-cd}
\usepackage{xcolor}
\usepackage{cancel}
\usepackage{amsmath}
\usepackage{slashed}
\usepackage{amssymb}
\usepackage{amsthm}
\usepackage{empheq}
\newtheorem{theorem}{Theorem}[section]
\newtheorem{prop}{Proposition}[section]

\theoremstyle{definition}

\theoremstyle{definition}

\theoremstyle{definition}
\newtheorem{mydef}{Definition}[section]
\theoremstyle{definition}

\usepackage{mathtools}
\usepackage{relsize}
\usepackage{xfrac} 
\usepackage{float}
\usepackage{braket}
\usepackage{multirow}
\usepackage{import}
\usepackage{xifthen}
\usepackage{pdfpages}
\usepackage{transparent}
\usepackage{setspace}
\usepackage[bookmarks]{hyperref}
\usepackage[nottoc]{tocbibind}
\usepackage{fancyhdr}

\usepackage[
backend=biber,
style=alphabetic,
sorting=ynt,maxbibnames=99
]{biblatex}
\usetikzlibrary{positioning}

\newcommand{\de}{\text{d}}
\newcommand{\tr}{\text{tr}}

\title{\textbf{
\vskip-3cm
Cobordism, Singularities and the Ricci Flow Conjecture}}
\date{}
\author[a]{David Martín Velázquez}
\author[a]{Davide De Biasio}
\author[a,b]{Dieter L\"ust}
\affil[a]{Max--Planck--Institut f\"ur Physik, Werner--Heisenberg--Institut, \newline
F\"ohringer Ring 6, 80805 M\"unchen, Germany}
\affil[b]{Arnold Sommerfeld Center for Theoretical Physics, \newline
Ludwig Maximilians Universit\"at M\"unchen, \newline  Theresienstrasse 37, 80333 M\"unchen, Germany}

\begin{document}
\fancypagestyle{plain}{%
	\fancyhead[R]{LMU-ASC 27/22 \\
MPP-2022-123}
	\renewcommand{\headrulewidth}{0pt}
}
\begin{spacing}{1.8}
   \maketitle 
\end{spacing}
\begin{abstract}
    In the following work, an attempt to conciliate the Ricci flow conjecture and the Cobordism conjecture, stated as refinements of the Swampland distance conjecture and of the No global symmetries conjecture respectively, is presented. This is done by starting from a suitable manifold with trivial cobordism class, applying surgery techniques to Ricci flow singularities and trivialising the cobordism class of one of the resulting connected components via the introduction of appropriate defects. The specific example of $\Omega^{SO}_4$ is studied in detail. A connection between the process of blowing up a point of a manifold and that of taking the connected sum of such with $\mathbb{CP}^n$ is explored. Hence, the problem of studying the Ricci flow of a $K3$ whose cobordism class is trivialised by the addition of $16$ copies of $\mathbb{CP}^2$ is tackled by applying both the techniques developed in the previous sections and the classification of singularities in terms of ADE groups.

\end{abstract}
\newpage
\tableofcontents
\newpage
  \section{Introduction}
The Swampland Program \cite{vafa2005string,Palti:2019pca,vanBeest:2021lhn,Brennan:2017rbf} originated as an endeavor to explore the idea that standard effective field theory techniques might break down when trying to incorporate gravitational physics into a low energy quantum field theory. This is expected to hold even at scales which lie way below Planck's mass 
\begin{equation}
    M_P\equiv\sqrt{\frac{\hbar c}{G}}\ ,
\end{equation}
where quantum gravity degrees of freedom are supposed to become relevant. A conspicuous body of evidence for the above-mentioned deviation from the typical effective field theory intuition, which suggests that quantum gravity phenomenology could safely be ignored in the low energy regime, has emerged from inspecting spatio-temporal and thermodynamical properties of black holes, general quantum gravitational arguments and explicit superstring theory constructions. Indeed, the belief that quantum gravity effects can strongly and further constraint the set of apparently consistent quantum field theories coupled to a dynamical space-time metric is now widely accepted and gets often expressed in the form of the so-called Swampland conjectures. In order for the discussion to acquire clarity, it is appropriate to introduce the notion of the Landscape, being the subset of low energy quantum field theories coupled to gravity that satisfy the restrictions imposed by Swampland conjectures and hence admit an ultraviolet completion to quantum gravity. By doing so, one straightforwardly defines the Swampland as the subset of apparently consistent low energy quantum field theories coupled to gravity that do not belong to the Landscape. The aim of the Swampland program is thus to find precise and unambiguous criteria that allow to classify theories into those belonging to the Landscape and those that fall into the Swampland. The natural prosecution of this analysis is to model such a set of theories, partitioned into the Landscape and the Swampland, as some sort of geometrically characterised moduli space, locally charted by the vacuum expectation values of some fields and provided with some notion of distance. The \textit{Swampland distance conjecture} was outlined in \cite{Ooguri:2006in} as the claim that large displacements in the moduli space should force an infinite tower of light states to descend into the spectrum of the low energy effective theory, lowering the energy cut-off $\Lambda_{\text{eff}}$ below which it functionally models relevant phenomena. This idea was pushed into conjecturing that infinite distance limits in the moduli space should involve significant alterations of the features of a low energy theory, such as the decompactification of some internal dimensions, the emergence of asymptotically tensionless and weakly coupled string degrees of freedom \cite{lee2022emergent} or the restoration of a global symmetry. The latter, in particular, connects to the long-standing idea that gravity should not allow for the existence of global symmetries \cite{banks1988constraints,banks2011symmetries,harlow2018constraints,harlow2021symmetries}, which is strongly backed up by semi-classical black hole physics, holographic arguments and by the fact that a global worldsheet symmetry in superstring theory always gets gauged in the target space-time. Both the \textit{Swampland distance conjecture} and the \textit{No global symmetries conjuecture} have been discussed and refined in many different instances. 
The \textit{Swampland distance conjecture} was first extendend to the case of AdS cosmological constant in \cite{lust2019ads}, showing how the infinite distance limit in which $\Lambda\to 0$ is accompanied by an infinite tower of light states. 
Extending the previously-introduced notion of a moduli space to regard the space-time metric components themselves as moduli, the most immediate way of displacing them is by considering paths induced by geometric flow equations \cite{Polyakov:1975rr,friedan1980nonlinear,friedan1985nonlinear}. Particular attention was dedicated to Ricci flow \cite{hamilton1982}
\begin{equation}
    \frac{\de g_{\mu\nu}}{\de\lambda}=-2R_{\mu\nu}\ ,
\end{equation}
which precisely reduces to the behaviour discussed in \cite{lust2019ads} in the case of AdS, and generalisations thereof. This way, exploiting mathematical methods which are rooted in string theory via the $\sigma$-model graviton renormalisation group flow, the \textit{Swampland distance conjecture} for the metric was rephrased and specified as the \textit{Ricci flow conjecture} \cite{Kehagias:2019akr,Bykov:2020llx,Luben:2020wix,DeBiasio:2020xkv}. The reader is strongly suggested to refer to \cite{Perelman:2006un,chow2004ricci,Cao2006ACP,morgan2007ricci,kleiner2008notes,Topping2006LecturesOT} for comprehensive discussions of the techniques that are usually employed when studying Ricci flow of differentiable manifolds. On the other hand, the recent \textit{Cobordism conjecture} \cite{mcnamara2019cobordism,ooguri2020cobordism,montero2021cobordism,dierigl2021swampland,buratti2021dynamical,andriot2022looking,angius2022end,blumenhagen2022dimensional,blumenhagen2022dynamical}, stating that viable quantum gravity backgrounds should have trivial cobordism class and thus be cobordant to one another, can arguably be regarded as a refinement of the \textit{No global symmetries} conjecture. It must be noted that the \textit{Ricci flow conjecture} typically deals with locally defined geometric flow equations and properties, while the \textit{Cobordism conjecture} comes as a genuine global statement. In the following work, an attempt to conciliate the two perspectives is presented. 
After having introduced some appropriate mathematical techniques, it is argued that the key mechanism via which Ricci flow can affect the global and topological properties of a manifold is for it to encounter a singularity, on which surgery procedures can be applied. These are expected to preserve cobordism classes, as they are typically enforced by substituting shrinking throats with smooth caps and connected sums are cobordant to disjoint unions. From a physics point of view, it is nevertheless usually pertinent to neglect one of the two connected components produced by a surgery. It is thus shown that defects must be introduced in order to trivialise the cobordism class of the resulting connected manifold and the specific example of $\Omega^{SO}_4$ is studied in detail. A connection between the process of blowing up a point of a manifold $\mathcal{M}$ and that of taking the connected sum of such manifold with $\mathbb{CP}^n$ is thereafter explored. Hence, the problem of studying the Ricci flow of a $K3$ whose cobordism class is trivialised by the addition of $16$ copies of $\mathbb{CP}^2$ is tackled both from a geometric flow perspective and by exploiting the classification of singularities in terms of ADE groups.

\section{Brief introduction to bordisms}
 At its core, the bordism relation is a simple one. We will begin by stating its simplest definition. Then, we will progressively add more structure, defining some of the concepts when necessary. Some other introductions are given in \cite{mcnamara2019cobordism, andriot2022looking}, as well as \cite[Chapters 4, 17]{milnor}. For a more technical and in-depth treatment of bordisms the reader is redirected to the classical reference, \cite{DanFreed}. 
 
 \subsection{Unoriented bordism}
 We now introduce the most basic bordism relation possible, where we do not take into account any kind of structure or special property of the bordant manifolds.
 \begin{mydef} Let $M$ and $N$ be two closed (i.e. compact and without boundary) $n$-dimensional manifolds. An (unoriented) cobordism between them is an $(n+1)$-dimensional manifold whose boundary is the disjoint union of $M$ and $N$, which we write as $\partial W=M\sqcup N$.
\end{mydef}

A more precise statement of the above discussion would involve the introduction of collar neighbourhoods on the boundary of the $(n+1)$-dimensional manifold, as well as a pair of maps defined on them such that, restricted over $\partial W$, they define embeddings of $M$ and $N$ into $W$. However, such precision would only be needed in formal proofs involving explicit gluings of the boundaries.

Because $M\sqcup N=N\sqcup M$, it is straightforward to check that cobordisms define an equivalence relation between manifolds. We denote this relation by $M\sim N$ without further reference to the type of cobordism which defines it, and we denote the equivalence class of a manifold $M$ by $[M]$. One may further endow the set of equivalence classes of closed $n$-manifolds with the structure of an abelian group by defining the sum of two classes as
\begin{equation}
[M]+[N]=[M\sqcup N].
\end{equation}
We denote this group by $\Omega_n$. Here, the zero-element with respect to the sum is taken to be the empty manifold $\emptyset$. In other words, we say that, given an $n$-dimensional maniofold, the fact that $[M]=0$ means that there exists some $(n+1)$-dimensional manifold $W$ such that $\partial W=M$.

Furthermore, we may consider a "multiplication", provided by the Cartesian product of manifolds $[M][N]\coloneqq[M\times N]$. This allows us to build up $\Omega_{*}\coloneqq\bigoplus_{n=0}^{\infty}\Omega_{n}$, and endow it with the structure of a graded ring. We call it the unoriented cobordism ring.

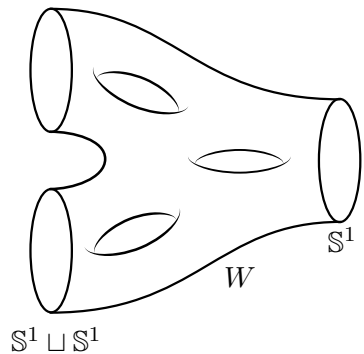
\begin{figure}[h!]
\centering
\def\svgwidth{0.35\columnwidth}
\begingroup%
  \makeatletter%
  \providecommand\color[2][]{%
    \errmessage{(Inkscape) Color is used for the text in Inkscape, but the package 'color.sty' is not loaded}%
    \renewcommand\color[2][]{}%
  }%
  \providecommand\transparent[1]{%
    \errmessage{(Inkscape) Transparency is used (non-zero) for the text in Inkscape, but the package 'transparent.sty' is not loaded}%
    \renewcommand\transparent[1]{}%
  }%
  \providecommand\rotatebox[2]{#2}%
  \newcommand*\fsize{\dimexpr\f@size pt\relax}%
  \newcommand*\lineheight[1]{\fontsize{\fsize}{#1\fsize}\selectfont}%
  \ifx\svgwidth\undefined%
    \setlength{\unitlength}{373.87222891bp}%
    \ifx\svgscale\undefined%
      \relax%
    \else%
      \setlength{\unitlength}{\unitlength * \real{\svgscale}}%
    \fi%
  \else%
    \setlength{\unitlength}{\svgwidth}%
  \fi%
  \global\let\svgwidth\undefined%
  \global\let\svgscale\undefined%
  \makeatother%
  \begin{picture}(1,0.82009391)%
    \lineheight{1}%
    \setlength\tabcolsep{0pt}%
    \put(0,0){\includegraphics[width=\unitlength,page=1]{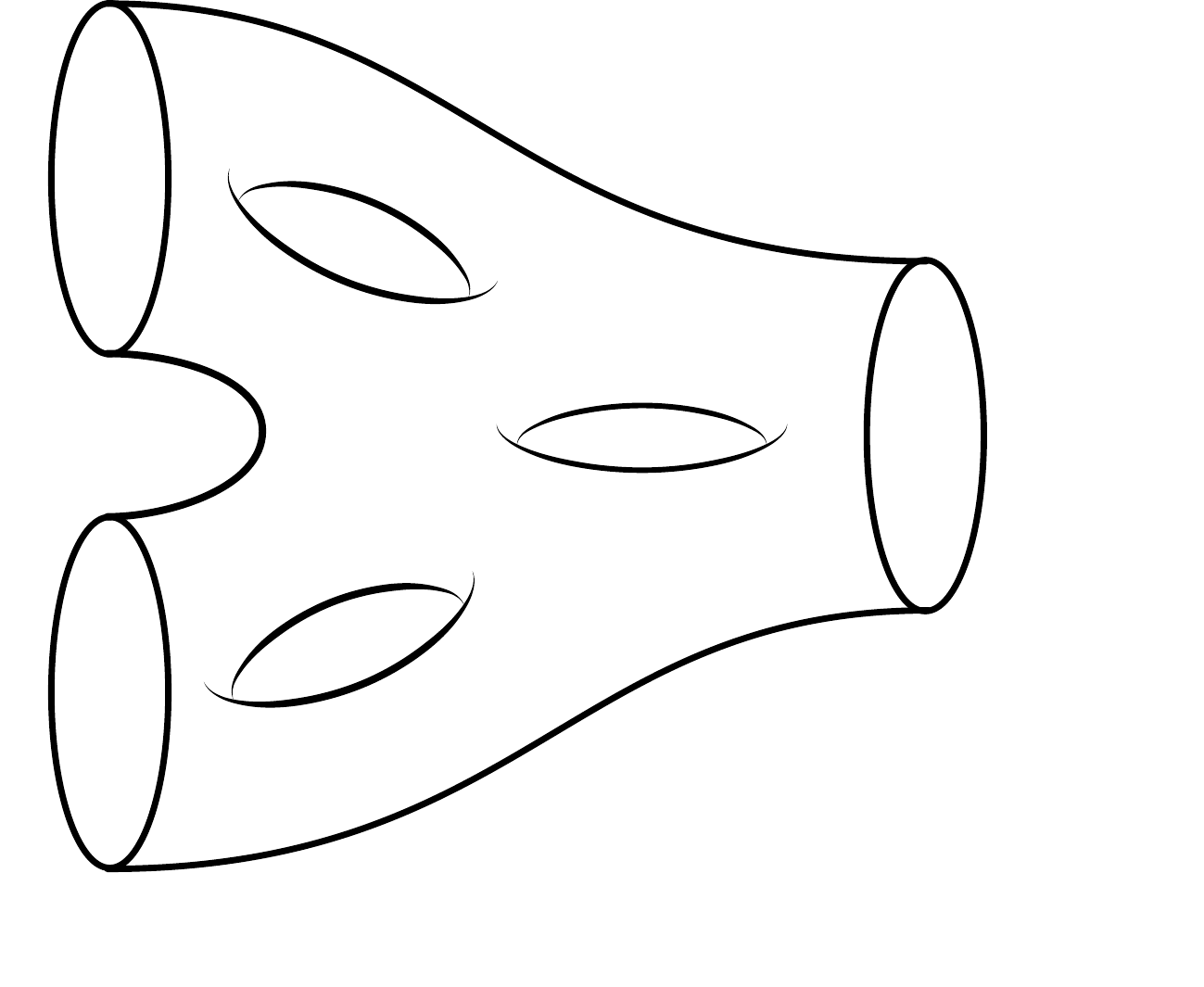}}%
    \put(0.49763304,0.15943824){\color[rgb]{0,0,0}\makebox(0,0)[lt]{\lineheight{1.25}\smash{\begin{tabular}[t]{l}$W$\end{tabular}}}}%
    \put(-0.00244225,0.00612002){\color[rgb]{0,0,0}\makebox(0,0)[lt]{\lineheight{1.25}\smash{\begin{tabular}[t]{l}$\mathbb{S}^1\sqcup\mathbb{S}^1$\end{tabular}}}}%
    \put(0.74122275,0.24254539){\color[rgb]{0,0,0}\makebox(0,0)[lt]{\lineheight{1.25}\smash{\begin{tabular}[t]{l}$\mathbb{S}^1$\end{tabular}}}}%
  \end{picture}%
\endgroup%

\label{fig:unorcob}
\caption{Simple example of an unoriented cobordism relation between two copies of the circle ($\mathbb{S}^1\sqcup\mathbb{S}^1$) and a single $\mathbb{S}^1$ by a surface of genus 3 with boundary.}
\end{figure}

A property that any unoriented cobordism ring is that all of its elements have characteristic 2, and thus it is an algebra over $\mathbb{Z}_{2}$. In other words, the grading of the ring is given by $\mathbb{Z}_2$, and all of its (nontrivial) cobordism classes are isomorphic to a $\mathbb{Z}_2$ factor.  This is easy to see by noting that, if we neglect orientations, $[M]+[M]=[M\sqcup M]$, and it is simple to construct a manifold whose boundary is $M\sqcup M$: it is $M\times [0,1]$. Hence, in the language of the group operation, we have that $2[M]=0$ for any manifold $M$.

\subsection{Adding more structure to the cobordism group}
While the definition we have made in the previous section is enough to understand the basic idea underlying the bordism relation, it falls flat when dealing with manifolds which allow for more physically interesting scenarios. For example, just to consider manifolds which have a definite orientation (and are, hence, orientable), we need to modify our original definition. The precise and rigorous way in which this is done is covered in \cite{DanFreed}, where tangential structures are introduced and treated in detail. A more condensed but readable account of this is contained in the introduction of \cite{andriot2022looking}.

For our present purposes, it is enough to know that, however the additional structure is defined on two $n$-dimensional bordant manifolds, it must extend in a coherent way to the $(n+1)$-dimensional manifold which connects them. This can be done by means of classifying spaces and bordisms with maps. A classifying space for a principal $G$-bundle is defined as follows. Given any Lie group $G$, we can define a space $BG$, together with a fibration given by its universal cover $EG\to BG$, with the defining property that any $G$-bundle over $X$ can be obtained as a pullback from $EG$. That is, given any $G$-bundle $P\to X$, there exists a map $f:X\to BG$ such that $P\simeq f^*EG$. This construction classifies $G$-bundles in the following sense. In general, given some $G$-bundle $E\to Y$, and two maps $f,g:X\to Y$, then the two pullback budles $f^*E$ and $g^*E$ over X are isomorphic if and only if $f$ and $g$ are homotopic (they belong to the same homotopy class). Therefore, at the level of the classifying space of a group, we have that all principal $G$-bundles over some manifold $X$ are classified by the homotopy classes of maps $X\to BG$. These classes of maps are often denoted by $[X,BG]$. 

There is an explicit construction of a classifying space for the structure groups of frame bundles, given by infinite Grassmannian manifolds. A Grassmannian manifold $G_n(\mathbb{R}^{n+k})$ is a manifold whose points are $n$-dimensional planes in $\mathbb{R}^{n+k}$.   Heuristically, one defines $BO(n)$ as the direct limit on $k$ of $G_n(\mathbb{R}^{n+k})$. In other words, one considers $k$ planes on higher and higher dimensional Euclidean spaces. WE highlight this because $BO(n)$ is the relevant classifying space for the case at hand, as all the structure groups we will consider are at most reductions of orthogonal ($O(n)$) bundles, which can be pulled back from the canonical bundle over an infinite Grassmannian.

While this notion is interesting in principle, we would like to have a practical way in which we could actually classify $G$-bundles over some space. Namely, we are looking for properties which allow us to tell bundles in different classes apart. This is where characteristic classes come into play. Because of their natural behavior under pullblacks, and the fact that they are invariant under isomorphisms, they are precisely the sort of tools to differentiate vector (and, in particular, principal) bundles over manifolds. For example, in the case of bordisms between orientable manifolds, the relevant characteristic classes are the Pontrjagin and Stiefel-Whitney classes, which give rise, respectively, to the Pontrjagin and Stiefel-Whitney numbers.

With classifying spaces, one can define certain structures on manifolds. For example, given some $n$-dimensional manifold $X$, a map $f:X\to BSO(n)$ induces an $SO$-bundle over $X$, i.e. an orientation. To make this fit with the bordism relation, we produce the following definition
\begin{mydef}
Let $X$ be a topological space. We denote by $(M,f)$ a pair consisting of a manifold $M$ and a continuous map $f:M\to X$. Now, we say that two pairs $(M,f)$, $(N,g)$ are bordant if 
\begin{itemize}
    \item There is a cobordism $W$ between $M$ and $N$ in the usual sense, and 
    \item There is a continuous map $F:W\to X$ such that 
    \begin{equation}
        F\mid_{M}=f, \hspace{1cm} F\mid_{N}=g,
    \end{equation}
    where $M$ and $N$ denote the inclusions of these manifolds into $\partial W$.
\end{itemize}
\end{mydef}
Let us consider an $SO$-structure as an example again. In the above definition, this implies taking\footnote{Note that we are now taking $BSO$ instead of $BSO(n)$. This is because we are juggling between two different dimensionalities. The classifying space $BSO$ is defined as a direct limit on $n$ of $BSO(n)$ spaces.} $X=BSO$. Then, the $SO$-structure on $W$, which is induced by $F:W\to BSO$ is such that it induces the appropriate orientation on the bordant manifolds, $M$ and $N$.

Once knowing how extra structure can be treated under the bordism relation in general, we now come back to the particular case of orientable manifolds, namely $SO$-structures. The most important result regarding this is that there are certain topological numbers which fully characterize the oriented bordism classes.
\begin{theorem}\label{thm:pnum}
Two manifolds are (oriented or unoriented) cobordant if and only if their Stiefel-Whitney and Pontryagin numbers are the same. In particular, a manifold $M$ is cobordant to $\emptyset$ if and only if all its Stiefel-Whitney and Pontryagin numbers vanish.
\end{theorem}

\section{Brief introduction to the Ricci flow equations}
The Ricci flow equations, first introduced by Richard Hamilton in \cite{hamilton1982} are a set of PDEs which modify the geometry of a given (differentiable) manifold. If $(\mathcal{M}, g(\lambda))$ is a family of Riemannian manifolds, parametrized by $\lambda$, we say that they are a solution to the Ricci flow equations if they satisfy 
\begin{equation}
    \frac{\de g_{\mu\nu}}{\de \lambda}=-2R_{\mu\nu},
\end{equation}
where $R_{\mu\nu}$ denotes the components of the Ricci tensor with respect to $g$. Originally, the flow equations were defined as a modification of some gradient flow equations, so as to make the solutions well-behaved. Indeed, the Ricci flow can be regarded as a sort of "heat equation" for manifolds (though the analogy is not totally rigorous). This is so because it tends to smooth out certain types of manifolds. In particular, its first use was to show that manifolds which admitted metrics of positive scalar curvature, also admitted metrics of \textbf{constant} positive scalar curvature. In recent years, it was famously used by Perelman to prove Thurston's geometrization conjecture, which contained the Poincaré conjecture as a corollary. A great introduction to this tool, which takes these recent developments, and covers the techniques developed by Perelman is \cite{Topping2006LecturesOT}. In this small section, we will not introduce such technicalities, but rather explore two simple scenarios for the Ricci flow, and then outline its importance in the context of the swampland program.

\subsection{The Ricci flow for Einstein manifolds}
An Einstein space is a Riemannian manifold $(\mathcal{M},g)$ for which the Ricci tensor is proportional to the metric, namely $R_{\mu\nu}=\Lambda g_{\mu\nu}$, with $\Lambda\in \mathbb{R}$. Some examples of this are the well-known maximally symmetric spaces: de Sitter ($\Lambda >0$), anti de Sitter ($\Lambda<0$), and Euclidean space $(\Lambda=0)$. Clearly, the Einstein condition greatly simplifies the Ricci flow equations, which are reduced to
\begin{equation}
\partial_{s}g_{\mu\nu}(s)=-2\Lambda g_{\mu\nu}
\end{equation}
The solution to this equation is a rescaling of the metric
\begin{equation}
g(s)_{\mu\nu}=(1-2\Lambda s)g(0)_{\mu\nu}
\end{equation}
This can be parametrized as a Weyl rescaling of the metric 
\begin{equation}
g(s)_{\mu\nu}=e^{-2\omega(s)}g(0)_{\mu\nu}.
\end{equation}

Alternatively, given that the solutions to the Ricci flow are given by rescaling, one is justified in condensing the flow equations by a single ODE for the cosmological constant. One can check\footnote{Though the calculations are rather cumbersome} that under the Ricci flow, the scalar curvature evolves as 
\begin{equation}
    \partial_sR=\nabla^2R+2R_{\mu\nu}R^{\mu\nu}.
\end{equation}
For an Einstein space, this reduces to the much simpler equation
\begin{equation}
    \Lambda'(s)=2\Lambda^2(s)
\end{equation}
Whose solution is
\begin{equation}
    \Lambda(s)=\frac{\Lambda_0}{1-2\Lambda_0(s-s_0)},
\end{equation}
where $\Lambda_0=\Lambda(s_0)$. Clearly, the behavior of the solution depends on the sign. In particular, we have that for a dS space, the solution shrinks to a point (the curvature approaches $\infty$) at finite time. On the other hand, an AdS space flows to flat space in the limit $s\to\infty$. Both of the perspectives we have seen can of course be related by means of a rescaling.

\subsection{Introducing Ricci solitons}
In the context of the Ricci flow, there is a further generalization of Einstein manifolds, which serves as motivation for the definition of new flows, such as the ones used by Perelman on his landmark works. So-called Ricci solitons are manifolds that change by a rescaling and a diffeomorphism under the flow. Let $(\mathcal{M}, g(s))$ be a Ricci flow, and define a one-parameter family of diffeomorphisms $\varphi_{s}:\mathcal{M}\to\mathcal{M}$, $\forall s\in [0,T)$ such that $\varphi_{0}=\mathds{1}$, as well as a scaling factor $\sigma(s)$ such that $\sigma(0)=1$. With these definitions, the condition for a manifold to be a Ricci soliton is phrased as
\begin{equation}
g(s)=\sigma(s)\varphi^{*}_{s}(g(0)),
\end{equation}
where $\varphi^{*}_{s}$ of course denotes the pullback given by the diffeomorphism. After differentiating the above expression, evaluating it at $t=0$, and using the explicit form of the Lie derivative defined by $V=\varphi'(s)$, then  we have that the PDE defining the flow of a Ricci soliton is given by
\begin{equation}
\frac{\partial g_{\mu\nu}}{\partial s}=\sigma'(0)g_{\mu\nu}(0)+\nabla_{\mu}V_{\nu}+\nabla_{\nu}V_{\mu}.
\end{equation}
Renaming $\sigma'(0)=-2\Lambda$, we obtain that the manifolds which display this sort of evolution under the Ricci flow must be such that
\begin{equation}\label{eq:RicSol}
-2R_{\mu\nu}=-2\Lambda g_{\mu\nu}+\nabla_{\mu}V_{\nu}+\nabla_{\nu}V_{\mu}.
\end{equation}
Indeed we can see that any Einstein manifold is a Ricci soliton, with $V_\mu=0$.

Examples of such manifolds are not so easy to come by as with Einstein manifolds. One such example of great physical signifiacnce is Witten's cigar, a two-dimensional geometry given by 
\begin{equation}
    g=\frac{1}{1+x^2+y^2}\left(dx^2+dy^2\right)
\end{equation}
For this metric, the Ricci tensor turns out to be 
\begin{equation}
    R_{\mu\nu}=\frac{2}{1+x^2+y^2}g_{\mu\nu}.
\end{equation}
This could in principle look like an Einstein manifold. However, if we define the radial vector field $Y=-2(x\partial_x+y\partial_y)$, we have that 
\begin{equation}
    \mathcal{L}_{Y}g=\frac{-4}{1+x^2+y^2}g
\end{equation}
Thus, the Ricci flow equations are sourced by a pure diffeomorphism, namely, $\Lambda=0$.


\subsection{The Ricci flow and the swampland}
Recent publications, as \cite{Kehagias:2019akr,Bykov:2020llx,Luben:2020wix,DeBiasio:2020xkv}, have recast the swampland distance conjecture in terms of the Ricci flow. To this end, a distance along the moduli space of metrics was defined, and then used to measure the distance in moduli space of metrics along the Ricci flow. It was then argued that an infinite distance along the Ricci flow shuould be accompanied by an infinite tower of states, whose mass scale is determined by said distance.

Of special importance to this refinement of the distance conjecture are those backgrounds which are fixed points of the flow, namely those which satisfy
\begin{equation}
    \frac{\partial g_{\mu\nu}}{\partial s}=0.
\end{equation}
Clearly, these are given by Ricci-flat manifolds. From the discussion above, one such example of a manifold flowing to flat space along the flow is $AdS$ space flowing to Euclidean space. From the perspective of the distance along the space of metrics, this is an incarnation of the $AdS$ distance conjecture. It states that, as $AdS$ space flows to flat space (namely, $\Lambda\to 0$), the associated tower of states with mass scale $m$ behaves in Planck units as $m\sim |\Lambda|^{\alpha}$, with $\alpha$ being a positive order one number.

In \cite{Kehagias:2019akr}, the Ricci flow distance conjecture was then refined in terms of Perelman's $\mathcal{F}$ and $\mathcal{W}$ entropy functionals (e.g. \cite[Chapter 8]{Topping2006LecturesOT}). Mathematically, the beauty of these functionals is that they allow for the consistent recasting of the Ricci flow equations in terms of gradient flows. Namely, one obtains the equations by means of the variational principle on certain functionals. Physically, this is useful for adding more fields to the effective theory, such as the dilaton field. One of most striking facts about the Perelman's entropy functionals is that they very closely resemble certain effective actions which are naturally obtained from the bosonic string theory. For example, the $\mathcal{F}$-functional has the form
\begin{equation}
    \mathcal{F}[g,f]=\int\left(R+\lvert\nabla f\rvert^2\right)e^{-f}dV
\end{equation}
where $f$ is a real scalar function on the background. This is closely related to the effective action in string theory with a metric and a dilaton field in the string frame. Indeed, the above $f$ plays the role of the dilaton, and we have that $f=2\phi$, and thus $g_s^2=e^f$, with $g_s$ being the string coupling constant.

The above observation is important in the context of the infinite distance conjecture, as it allows us to relate the Ricci flow equations (obtained form $\mathcal{F}$) with the RG flow equations form the string effective action. This plays into the fact that the distance conjecture, at its core, studies the limits where certain theories are strongly and weakly coupled. In this context, we can also use the $\mathcal{F}$ and $\mathcal{W}$ functionals themselves as a measure of the distance along the flow. This is basically because the flow equations are obtained by finding the maximal rate of change of the functionals, and thus these are monotonic along the flow.

\section{The cobordism class and Ricci flow for $\Omega_{4}^{SO}$}

Our first task is to give solid arguments to justify the behavior of cobordism classes of manifolds under the Ricci flow. Intuitively, given the smooth nature of the flow equations, one expects the cobordism class of the initial manifold to remain unchanged wherever the solution is defined. In other words, we expect the flow not to change the cobordism class, at least unless we encounter a singularity. Indeed, given a family of manifolds $(\mathcal{M}, g(s))$ parametrized by the flow time $s$, let $T$ be the maximal time for which the solution is defined. Then for every $T'<T$, we expect to be able to construct a cobordism between $(\mathcal{M},g(0))$ and $(\mathcal{M}, g(T'))$ by following the flow.

In this section we will  show that this is the case for Einstein spaces and for the more general Ricci solitons. For it, we will need a quantity which classifies the cobordism class of a manifold, and that can be related to the Ricci Flow in some way. Theorem \ref{thm:pnum} gives us such a quantity: the Pontrjagin numbers of the frame bundle of the manifold. Indeed, if we were to flow some manifold, and find that at some value of the flow time one of its Pontrjagin numbers has changed, then it would follow that it belongs to a different cobordism class than the original manifold.

Because $p_{1}(\mathcal{R})\in H^{4j}(\mathcal{M},\mathbb{Z})$, the smallest dimension for which $p_{1}(\mathcal{R})$ does not trivially vanish is $d=4$, where the first Pontrjagin class is the only one that does not automatically vanish. In this case, we have that 
\begin{equation}
p_{1}(\mathcal{R})=-\frac{1}{8\pi^{2}}\tr\mathcal{R}^{2},
\end{equation}
where $\mathcal{R}$ of course stands for the curvature 2-form of the frame bundle. Consequently, the trace is defined in terms of the $GL(4,\mathbb{R})$ indices of the curvature, resulting in a top form in $d=4$. There is, however, one subtlety with the above definition, and it is that $\mathcal{R}$ is defined in terms of a coordinate frame of the bundle, while the Ricci flow is more naturally defined in coordinates. In order to make the distinction explicit, we use latin indices to denote frame indices, and greek indices otherwise. We can thus write the above definition in a more explicit fashion as 
\begin{equation}
p_{1}(\mathcal{R})=-\frac{1}{8\pi^{2}}{\mathcal{R}^{a}}_{b}{\mathcal{R}^{b}}_{a}.
\end{equation}

The curvature 2-form can be defined in terms of the Riemann tensor as follows
\begin{equation}
{\mathcal{R}^{a}}_{b}={R^{a}}_{bcd}\hat{\theta}^{c}\wedge\hat{\theta}^{d},
\end{equation} 
where $\lbrace \hat{\theta}^{a}\rbrace$ is the orthonormal coframe. We thus have that 
\begin{align*}
p_{1}(\mathcal{R})&={R^{a}}_{bcd}{R^{b}}_{aef}\hat{\theta}^{c}\wedge\hat{\theta}^{d}\wedge\hat{\theta}^{e}\wedge\hat{\theta}^{f} \\
&={R^{a}}_{bcd}{R^{b}}_{aef}\varepsilon^{cdef}d\mu,
\end{align*}
where $d\mu=\hat{\theta}^{c}\wedge\hat{\theta}^{d}\wedge\hat{\theta}^{e}\wedge\hat{\theta}^{f}=\sqrt{g}d^{4}x$. As we have already said before, the above equation for the first Pontrjagin class is expressed with respect to an orthonormal frame (i.e. a noncoordinate basis). Therefore, we must obtain its coordinate form. We can do this by means of the \textit{vielbeins} $\lbrace{e_{\mu}}^{a}\rbrace\subset GL(4,\mathbb{R})$. There are quite a few indices to transform, but we see that
\begin{equation}
{R^{a}}_{bcd}{R^{b}}_{aef}={R^{\alpha}}_{\beta\mu\nu}{R^{\beta}}_{\alpha\rho\sigma}{e_{\alpha}}^{a}{e^{\beta}}_{b}{e_{\beta}}^{b}{e^{\alpha}}_{a}{e^{\mu}}_{c}{e^{\nu}}_{d}{e^{\rho}}_{e}{e^{\sigma}}_{f}
\end{equation}
Due to the original trace on $\mathcal{R}$, the first four \textit{vielbeins} turn out to amount to two traces of $g_{\alpha\beta}$, and hence are constant. Because we will only be concerned with the variation under the Ricci Flow of $p_{1}$, we may drop constant prefactors and write
\begin{equation}
p_{1}(\mathcal{R})\sim{R^{\alpha}}_{\beta\mu\nu}{R^{\beta}}_{\alpha\rho\sigma}{e^{\mu}}_{c}{e^{\nu}}_{d}{e^{\rho}}_{e}{e^{\sigma}}_{f}\varepsilon^{cdef}\sqrt{g}d^{4}x.
\end{equation}

\subsection{Ricci flow and Einstein manifolds in $\Omega_4^{SO}$}

For an Einstein space, we know that $R_{\mu\nu}=\Lambda g_{\mu\nu}$, and therefore the Ricci Flow equations are reduced to 
\begin{equation}
\partial_{s}g_{\mu\nu}(s)=-2\Lambda g_{\mu\nu}
\end{equation}
The solution to this equation is a rescaling of the metric
\begin{equation}
g(s)_{\mu\nu}=(1-2\Lambda s)g(0)_{\mu\nu}
\end{equation}
This can be parametrized as a Weyl rescaling of the metric 
\begin{equation}
g(s)_{\mu\nu}=e^{-2\omega(s)}g(0)_{\mu\nu}.
\end{equation}
In this situation, seeing that $p_{1}(\mathcal{R})$ is invariant is fairly simple. Under a rescaling, the relevant quantities transform as
\begin{equation}
g_{\mu\nu}\mapsto\kappa g_{\mu\nu}\Rightarrow\begin{cases}
{R^{\alpha}}_{\beta\mu\nu}\mapsto{R^{\alpha}}_{\beta\mu\nu} \\
d\mu\mapsto \kappa^{d/2}d\mu \\
{e^{\mu}}_{c}\mapsto \kappa^{-1/2}{e^{\mu}}_{c}
\end{cases}
\end{equation}
The second property comes from the appeareance of $\sqrt{g}$ in the measure, and the third from the fact that $g_{\mu\nu}={e_{\mu}}^{c}e_{\nu c}$. We conclude that the first Pontrjagin number is invariant under the Ricci flow for Einstein manifolds (this result is clearly seen to extend to manifolds of dimension $d=4n$). 

\subsection{Ricci flow and Ricci solitons}
We now come to check the invariance of the characteristic numbers of an orientable 4-dimensional manifold under the Ricci flow. Instinctively, we expect that they should be, as the only thing that remains is to study their behavior under a diffeomorphism. 

To prove the previous statement, begin by considering some diffeomorphism\footnote{We are using different manifolds for the sake of clarity of the explanation.} $\varphi:\mathcal{M}\to\mathcal{N}$. Not only does such a map define an isomorphism of the homology and cohomology of both manifolds, but also allows us to define the pullback bundle $\varphi^*T\mathcal{N}$. Due to the naturality of Stiefel-Whitney and Pontrjagin classes, we have that the characteristic classes over $\mathcal{M}$ are the pullback of the ones over $\mathcal{N}$, and thus, $p_{i}(\mathcal{M})=\varphi^*p_{i}(\mathcal{N})$, and analogously with Stiefel-Whitney classes. Furthermore, in the context of smooth manifolds these are nothing but differential forms. It is a standard result \cite{lee} that under these conditions, and given some generic top form $\omega$ over $N$, 
\begin{equation}
\int_{\mathcal{N}}\omega=\pm\int_{\mathcal{M}}\varphi^*\omega,
\end{equation}
where the sign on the right hand side depends on if the diffeomorphism is orientation preserving or orientation reversing. Now, in our particular case, for the one-parameter family of diffeomorphisms, $\mathcal{M}=\mathcal{N}$. Moreover, we know that $\varphi_{0}=\mathds{1}$, and therefore for every $s\in [0,T)$, $\varphi_{s}:\mathcal{M}\to\mathcal{M}$ is continuously connected to the identity, making it orientation preserving. Thus, the invariance of the Pontjagin and Stiefel-Whitney numbers under a flow defined by \eqref{eq:RicSol} follows.

\subsection{Pinching of necks}\label{sect:necks}
One of the phenomena which is tightly linked with the Ricci flow is the appeareance, and the subsequent pinching-off, of necks on a manifold. When a neck appears, it is very likely that it will shrink faster, compared to other parts of the manifold. This leads us to the conclusion that, at some point along the flow, a section of the manifold (located on the neck), will shrink to a point, thus reaching the end of the flow. However, the flow equations could still be well defined at all other parts of the manifold. Because of this, it would be interesting to have a mechanism to allow us to study the behavior of the manifold under the flow \textit{past} the pinching of these necks, that would reproduce the intuitive behavior of a neck pinching off, while at the same time changing the geometry as little as possible. We will now see that this procedure is compatible with the bordism relation. First, we will discuss the notion of a \textit{surgery} on a manifold, and then we will explain a way to adapt the notion of convergence in order to allow for such a modification to take place.

\begin{figure}[t!]
\centering
\def\svgwidth{0.65\columnwidth}
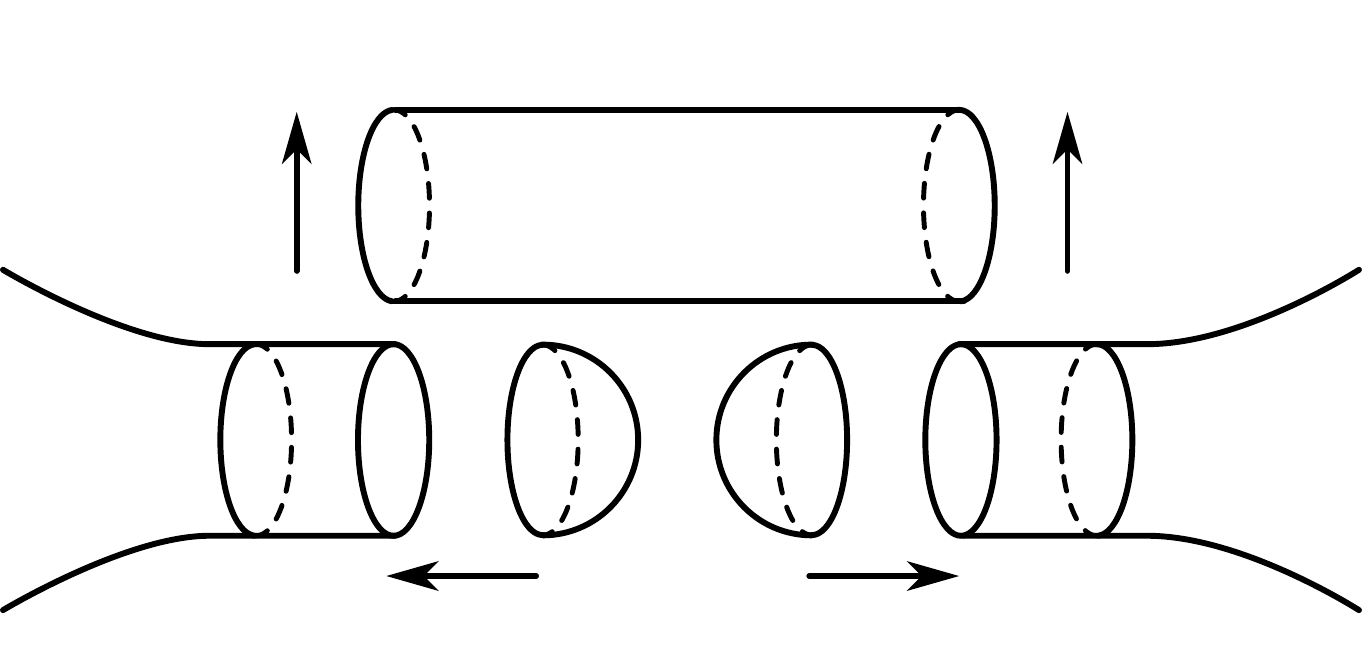
\label{fig:holo}
\caption{Schematic of the surgery that is to be done when a neckpinch ocurs.}
\end{figure}
In mathematical terms, a topological neck $\mathcal{N}$ in a manifold $\mathcal{M}$ is a local diffeomorphism of a cylinder into $\mathcal{M}$:
\begin{equation}\label{eq:neck}
\mathcal{N}:\mathbb{S}^{n-1}\times [a,b]\to \mathcal{M}.
\end{equation}
The treatment of necks in the context of the Ricci flow is a very delicate one. A very detailed description of it is given in one of Hamilton's original papers, \cite{Hamilton4Man}. For our present purposes, it is only necessary to know that the strategy to follow when approaching a singularity in the form of a neck pinch is to perform a so-called surgery on the manifold. A surgery is a topological procedure which is based on the observation that the relation 
\begin{equation}
    \partial(\mathcal{M}\times \mathcal{N})=\partial\mathcal{M}\times \mathcal{N}\cup \mathcal{M}\times \partial\mathcal{N}
\end{equation}
Implies the two following possibilities
\begin{equation*}
    \mathbb{S}^p\times \mathbb{S}^{q-1}=\begin{cases}
    \partial\left(D^{p+1}\times \mathbb{S}^{q-1}\right) \\
    \partial\left(\mathbb{S}^p\times D^q\right)
    \end{cases}
\end{equation*}
Since the two manifolds on the right hand side have the same boundary, one could in principle cut one out, and glue the other one along the leftover boundary. This is precisely a surgery. In mathematical terms, we have the following
\begin{mydef}\label{def:surgery}
Let $\mathcal{M}$ be a manifold of dimension $n=p+q$, and let $\varphi:\mathbb{S}^p\times D^q$ be an embedding. Define another manifold $\mathcal{M}'$ as 
\begin{equation*}
    \mathcal{M}'=\left[\mathcal{M}\setminus \text{Int}(\text{Im}(\varphi))\right]\cup_{\varphi}\left[D^{p+1}\times \mathbb{S}^{q-1}\right].
\end{equation*}
We define the surgery to be the triple $(\mathcal{M}, \mathcal{M}', \varphi)$\footnote{Note that we have not specified \textit{on which} manifold we perform the surgery. This is because surgeries define an equivalence relation, and are in particular reflexive}.
\end{mydef}
In this way, looking at \eqref{eq:neck}, if we perform a surgery, we substitute the pinched neck by two smooth caps, and then keep flowing the resulting manifold. In a sense, this procedure realizes the inverse of a disconnected sum, and thus gives a connected sum decomposition of the original manifold. While this part is rather simple, in the context of the Ricci flow we are interested not only on the topology, but also on the geometry of the manifold. Hence, the problem arises when one tries to extend the metric smoothly from the bases of the neck to the smooth caps. A detailed account of this is contained in \cite{Hamilton4Man}.

What we have defined above is a general surgery, of which the substitution of a neck for two smooth caps is a particular example of. From the definition of a neck in \eqref{eq:neck}, we need only choose $p=n-1$ and $q=1$. In fact, the reader might have noticed that this particular example of a surgery is precisely a connected sum in reverse. Indeed, one of the topological appeals of the Ricci flow is that it provides a connected sum decomposition of the manifolds on which it is defined. This is very convenient for our purposes, as we know that the process of taking a connected sum defines a bordism $\mathcal{M}\sqcup\mathcal{N}\sim\mathcal{M}\#\mathcal{N}$. In general, any type of surgery in the sense of definition \ref{def:surgery} defines a bordism relation from $\mathcal{M}$ to $\mathcal{M}'$, by means of taking the so-called \textit{trace of the surgery}, defined as follows. In the conditions of definition \ref{def:surgery}, we define an $n+1$ dimensional manifold $\mathcal{W}$ by first taking $\mathcal{M}\times I$, where $I$ is the unit interval, and then attaching $D^{p+1}\times D^q$ along one of the "ends" of $\mathcal{M}\times I$. All together, we define 
\begin{equation}
    \mathcal{W}=\left[\mathcal{M}\times I\right]\cup_{\varphi\times \{1\}}\left[D^{p+1}\times D^q\right].
\end{equation}
It can be checked that this manifold indeed satisfies $\partial\mathcal{W}=\mathcal{M}\sqcup\mathcal{M}'$, by recalling the gluing procedure, the definition of $\mathcal{M}'$, and noting that one of the two boundaries of $D^{p+1}\times D^q$ is precisely the one that is glued to $\mathcal{M}'$. All in all, what we have learned form this discussion is that the cobordism class of a manifold is left unchanged by the Ricci flow, even if necks pinch off, due to the fact that the manifolds before and after the pinching are bordant.

Finally, there is one more aspect about neck pinching that we should look at, for the sake of completeness. Since we have just defined a way to extend the flow past a certain type of singularity, which in particular involves stopping the flow whenever a neck forms, only to turn it back on after the surgery, we need to define our notion of convergence. To add to that, this notion of convergence should be diffeomorphism invariant, so a notion of convergence of metric as tensors is unsatisfactory. Indeed, we need to specify what we mean when we say that some manifold $\mathcal{M}$ converges to some other manifold $\mathcal{M}'$ (or has $\mathcal{M}'$ as its limit). This is done by means of the smooth, pointed \textit{Cheeger-Gromov} convergence of manifolds. Here, the term \textit{pointed} means that the manifolds which compose the convergent sequence have a certain choice of distinguished points. The proper defintion is as follows
\begin{mydef}
A sequence of smooth, complete, and pointed Riemannian manifolds $(\mathcal{M}_i,g_i,p_i)$, with $p_i\in \mathcal{M}_i$ is said to converge smoothly to the smooth, complete, and pointed manifold $(\mathcal{M},g,p)$ as $i\to\infty$ if there exist
\begin{itemize}
    \item a sequence of compact sets $\{\Omega_i\}$, with $\Omega_i\subseteq \mathcal{M}$ exhausting $\mathcal{M}$ with $p\in \text{Int}(\Omega_i)$ for each $i$, and 
    
    \item  sequence of smooth maps $\varphi_i:\Omega_i\to \mathcal{M}_i$ which are diffeomorphic onto their image and satisfy $\varphi_i(p)=p_i$ for all $i$,
\end{itemize}
such that $\lim_{i\to \infty}\varphi^*_ig_i\to g$ smoothly in the sense that for all compact sets $K\subseteq \mathcal{M}$, the tensor $\varphi^*_ig_i-g$ and its covariant derivatives of all orders converge uniformly to 0 on $K$
\end{mydef}
This definition was taken from \cite{Topping2006LecturesOT}. The interested reader is directed to this reference for more details on this and almost all basic concepts relative to the Ricci flow.

\section{The main case study: $\Omega^{SO}_4$}
We introduce the Pontrjagin and Stiefel-Whitney numbers as the relevant topological invariants for the $SO$ structure. For the particular case of the 4-dimensional group, we discuss how the fact that the signatue $\sigma$ of $C\mathbf{P}^2$ is $\sigma(C\mathbb{P}^2)=1$ can be exploited to kill cobordism classes by the addition of more and more copies of the complex projective space. In particular, 16 copies of $C\mathbb{P}^2$ may be used to kill the class of $K3$.

After having killed the class of $K3$, we must give a physical (or otherwise) interpretation of the complex projective planes. One such interpretation is given in terms of the blow-up of singularities of the $K3$ surface. These are classified in terms of the ADE groups. We will explore this in detail in later sections.

Another interpretation is given by \cite{nCP2s}, where a metric on the connected sum of copies of $C\mathbb{P}^2$ is provided, together with the interpretation of the metric of magnetic monopoles in hyperbolic space. Upon investigation, the given metric is strikingly similar to the Kaluza-Klein monopole in 4 dimensions (also called the Taub-NUT geometry). Indeed, the requirements on the connection over the circle bundle present in the monopole coincide in the two cases.


\subsection{The relevant invariants}
We will now study the particular case of $\Omega_4^{SO}\simeq\mathbb{Z}$. Since the group is nontrivial, we know that the cobordism conjecture forces us to consider more restricted structures or add defects to the theory. We now cite two results which will be of use, both of which can be found in \cite{milnor}. The first is that, in the context of oriented bordism (namely, an $SO$-structure), the relevant bordism invariants are the Pontrjagin and Stiefel-Whitney numbers, which we denote for some manifold $M$ by $p_j(M)$ and $w_j(M)$. Another result is the Hirzebruch signature theorem, which states that the signature of a manifold can be written in terms of the Pontrjagin numbers and is thus, by the previous result, a cobordism invariant.

The four dimensional case is special in the above sense, as the Pontrjagin classes $p_j$ are cohomology classes in $H^{4j}(M,\mathbb{Z})$. Therefore $d=4$ is the first dimension for which we have a class which is \textit{a priori} nontrivial\footnote{Note that $p_0$ is axiomatically set to 1.}. Therefore, we will use the Pontrjagin numbers to label cobordism classes, and thus construct manifolds such that $p_1(M)=0$. We will see that in the examples we construct, this will also force all Stiefel-Whitney numbers to be 0, thus implying that their cobordism class is indeed trivial. As we have said in the introduction, the trivialization of these classes will also mean that for these examples, the structure group can be further reduced. However, we will be more interested in the perspective of adding defects.

\subsection{The case for $\mathbb{CP}^{2}$}

It is a known result \cite{milnor} that $\mathbb{CP}^{2}$ is the generator of $\Omega_{4}^{SO}$, the oriented cobordism group of dimension 4. Furthermore, Hirzebruch's signature theorem states that
\begin{equation}
p_{1}(\mathbb{CP}^{2})=3\Leftrightarrow \sigma(\mathbb{CP}^{2})=\frac{p_{1}(\mathbb{CP}^{2})}{3}=1.
\end{equation}
Since the signature is a cobordism invariant, we assign $[\mathbb{CP}^{2}]=1$ (using additive notation). Constructing a background with trivial cobordism class is now straightforward: we take two copies of $\mathbb{CP}^{2}$ and reverse the orientation in one of them. Thus, the background 
\begin{equation}
\mathcal{M}=\mathbb{CP}^{2}\sqcup\overline{\mathbb{CP}}^{2}
\end{equation}
bounds. Indeed, we have that, by the properties of Pontrjagin numbers
\begin{equation}
\sigma(\mathcal{M})=\sigma(\mathbb{CP}^{2})+\sigma(\overline{\mathbb{CP}}^{2})=1-1=0 
\end{equation}
so we conclude that $[\mathcal{M}]=0$.

As an additional note, we can use the fact that the disjoint union of two manifolds is cobordant to their connected sum, and then use a long exact sequence of a pair and a Mayer-Vietoris argument to show that
\begin{equation}\label{eq:hom2cp2}
H_{p}(\mathbb{CP}^{2}\#\overline{\mathbb{CP}}^{2})=
\left.\begin{cases}
\mathbb{Z} \hspace{1.5cm}& p=0,4 \\
\mathbb{Z}\oplus\mathbb{Z} &p=2 \\
0 &\text{else}
\end{cases}\right\}=H^{p}(\mathbb{CP}^{2}\#\overline{\mathbb{CP}}^{2})
\end{equation}
This tells us that the manifold we have just constructed is not homeomorphic to $\mathbb{S}^{4}$, even if cobordant to it.

\subsection{The case for $K3$}
One of the manifolds that gives rise to an interesting scenario in our study of cobordism classes and the Ricci flow is $K3$. In \cite{mcnamara2019cobordism}, the authors have already studied this manifold from the point of view of its cobordism class. For our present purposes, an important result that is cited there is that 
\begin{equation}
    \sigma(K3)=-16
\end{equation}
Again, since the signature of a manifold is a cobordism invariant, this defines a class in the cobordism group $\Omega_4^{SO}$. With $\mathbb{CP}^2$ as a generator, we can write this as 
\begin{equation}
    [K3]=-16[\mathbb{CP}^2]
\end{equation}
where, of course, we use the notation of addition to mean the disjoint union of 16 projective planes with their orientation reversed. From this, we can easily construct a manifold which has a trivial cobordism class, given by 
\begin{equation}
    K3\sqcup 16\mathbb{CP}^2,
\end{equation}
which by the above argument has vanishing signature and hence bounds a $5$-manifold.

The interesting part about the above scenario is that, if we turn on the regular Ricci flow equations, all of the $\mathbb{CP}^2$ (being that they are Einstein spaces with positive sectional curvature) shrink to a point. However, $K3$ is a Calabi-Yau manifold, and in particular Ricci-flat. Therefore, it is a fixed point of the flow equations, and remains invariant under them. Furthermore, because the disjoint union and the connected sum of manifolds are cobordant, it is irrelevant if we consider the above manifold as $K3\sqcup \mathbb{CP}^2\sqcup\ldots\sqcup\mathbb{CP}^2$ or as $K3\# \mathbb{CP}^2\#\ldots\#\mathbb{CP}^2$. Therefore, we do not need to worry about the necks that are formed by the connected sums. It follows that in this setup, at the end of the flow, what remains is the original $K3$, and all of the remaining copies of projective planes have shrunk to points. This loose discussion represents the idea that there is some connection, due to working under the bordism relation, of trivializing the bordism class at the topological level, and adding defects to the manifold. The connection will be given by the blow-ups (minimal resolutions) of the $K3$ singularities.




\section{Blowing up points and connected sums of copies of $\mathbb{CP}^2$}
In the context of algebraic geometry, there is a procedure known as "blowing up". The underlying idea is to substitute a certain subspace of some total space (usually some complex variety) by another, more well-behaved space. Furthermore, one does it in such a way that the resulting total space looks like the one we started with away from the blown-up subspace. While this procedure can be done just for the sake of it, it fully comes into fruition in the context of the resolution of singularities, where there is some sort of singular subspace (e.g. a conifold point), which is then substituted by a smooth space, blended into the total space in such a way that the combination of the two is smooth throughout. In this section, we will focus on the case where the singular subspace is just a singular point, we will give the explicit construction of its blow up, and we will see how it relates to manifolds containing connected sums of complex projective spaces.

\subsection{The blow up of a point}

We will begin by considering $\mathbb{C}^2$. The only generality that is lost in doing so is that of dimensionality, which will not be much of a problem, as the construction can be easily generalized to higher dimension. Since all of the considerations we will make will be local in nature, we can always reduce the case of an arbitrary complex manifold of dimension $n$ to that of $\mathbb{C}^n$. Now, suppose we wanted to blow up the origin $\mathcal{O}\in\mathbb{C}^2$. This would produce a new space, which we shall name $Bl_{\mathcal{O}}(\mathbb{C}^2)$, and define as follows. Consider the set of all paths in $\mathbb{C}^2$ whose endpoint is $\mathcal{O}$, and for each of them, take their tangent line at the origin. This set of lines through the origin is precisely the definition of $\mathbb{CP}^1$, each line $l$ is given by a point $[y_1:y_2]\in\mathbb{CP}^1$ (where we have used homogeneous coordinates) and vice versa. Define $Bl_{\mathcal{O}}(\mathbb{C}^2)\subseteq \mathbb{C}^2\times \mathbb{CP}^1=\{(x_1,x_2,[y_1:y_2])|(x_1,x_2)\in\mathbb{C}^2, [y_1:y_2]\in\mathbb{CP}^1\}$ as the zero locus of 
\begin{equation}\label{eq:blowup}
    p=x_1y_2-x_2y_1,
\end{equation}
namely, all of the points verifying $x_1y_2=x_2y_1$. This condition is equivalent to saying that the point $(x_1,x_2)$ is contained in the line $l\in\mathbb{CP}^1$, since the above equations are obtained by imposing linear dependence of the set of vectors defined by 
\begin{equation}
    \begin{pmatrix}
    x_1 & x_2 \\
    y_1 & y_2
    \end{pmatrix},
\end{equation}
by requiring that the rank of the matrix be 1.

Let us look at \eqref{eq:blowup} more closely. Whenever $(x_1,x_2)\neq (0,0)$, then $(y_1,y_2)$ are completely determined up to scaling, which means that they are actually completely determined in $\mathbb{CP}^1$. On the other hand, at the origin, $y_1$ and $y_2$ can take any value, and \eqref{eq:blowup} vanishes tivially, so both coordinates range over all of $\mathbb{CP}^1$. These considerations lead us to define the following projection map
\begin{equation}
    \begin{split}
        \pi:\hspace{0.2cm}\mathbb{C}^2\times\mathbb{CP}^1 \hspace{0.2cm} &\to \hspace{0.3cm}\mathbb{C}^2 \\
        (x_1,x_2,[y_1:y_2])&\mapsto (x_1,x_2)
    \end{split}
\end{equation}
By the above discussion, this map is an isomorphism away from the origin, but at the origin itself we have that $\pi^{-1}(0,0)=\mathbb{CP}^1$.

\subsection{Blowing up and connected sums}
In this section, we will try to draw a connection between the process of blowing up a point and that of taking the connected sum of the original manifold with $\mathbb{CP}^n$. This is encapsulated in the following proposition:
\begin{prop}
Let $x\in\mathcal{M}$ be a point in a complex $n$-dimensional manifold $\mathcal{M}$. Then the blow up $Bl_x(\mathcal{M})$ is diffeomorphic as an oriented manifold to $\mathcal{M}\#\overline{\mathbb{CP}^n}$.
\end{prop}

\begin{proof}
Because the blow up of a point is a local transformation, we can set, without loss of generality, $\mathcal{M}=D^n=\{z\in\mathbb{C}^n| \lVert z\rVert^2<1\}$, and $x=0$. By the definition introduced in the preceding section, the blow up is given by
\begin{equation}\label{eq:blowup2}
    Bl_0(D)=\{(x,z)\in D\times\mathbb{CP}^{n-1}| x_iz_j=x_jz_i, \ \forall \ i\neq j\}
\end{equation}

Now we have to construct the connected sum. To do so, we must first embed two disks $D^n$ in our original manifolds, quotient out by their images, and then glue along the remaining boundaries. Because we have set $X=D^n$, one of the embeddings can just be taken as the identity $\text{id}:D^n\to X$. For the remaining one, in order to maintain the overall orientation of the manifold, we must choose an orientation reversing map $f:D^n\to\mathbb{CP}^n$. We thus define it to be $x\mapsto[1:x]$. The fact that this is indeed an orientation reversing map can be seen by considering inward pointing rays on each of the coordinates of $D^n$. This allows us to restrict the image rays to $\mathbb{CP}^1\subset \mathbb{CP}^n$ by setting the remaining (homogeneous) coordinates to 0. Then, since $\mathbb{CP}^1\simeq \mathbb{S}^2$, and $f$ can be seen locally as an inversion ($z\mapsto 1/z$) at the level of the regular coordinates obtained from the homogeneous ones, we see that the orientation of the rays is reversed.

Once we have the embedding maps, we consider the original manifolds and remove an annulus inside the embedded disk. The fact that it is an annulus and not a proper disk is topologically irrelevant. For $X$ this is trivial, and for $\mathbb{CP}^n$ this is given by
\begin{equation}\label{eq:cpdisk}
     \mathbb{CP}^n\setminus f\left(\frac{1}{2}\overline{D}\right)\coloneqq \{[x_0:x] \mid |x_0|<2\lVert x\rVert  \}
\end{equation}
Note that we are using homogeneous coordinates to describe the manifold, and hence the particular numerical value of $x_0$ is irrelevant, and the removed disk can be encapsulated in the restriction on the right hand side of \eqref{eq:cpdisk}. All that remains to show our desired result is to construct the gluing diffeomorphism (so that we can define $X\#\mathbb{CP}^n$), and a diffeomorphism from the connected sum to $Bl_0(D)$. To do so, we will kill two birds with one stone, and use the fact that $Bl_0(D)$ is isomorphic to $D$ away from the origin. Hence, the gluing map from $\mathbb{CP}^n$ to $D$ will be the same as the one to $Bl_0(D)$ (recall that we have chosen $X=D$). We can give an explicit orientation preserving map which serves this purpose:
\begin{align*}
    \hat{\xi}:\mathbb{CP}^n\setminus f\left(\frac{1}{2}\overline{D}\right)&\to Bl_0(D) \\
    [x_0:x]&\mapsto \left(\frac{x_0}{2\lVert x\rVert^2}x,[x]\right)
\end{align*}
We will first check that it restricts to the gluing map on the embedded disk, namely on $f(D)\setminus f(\frac{1}{2}D)$. In the above notation, this means that we must set $x_0=1$, and let $\lVert x\rVert < 1$ (since this is the coordinate on the disk). Both of these conditions restrict the values of $x$ to be on the annulus $\frac{1}{2}<\lVert x\rVert <1$. Therefore, $x$ itself is defined on the annulus (as opposed to the embedded annulus on $\mathbb{CP}^2$). At the level of $\hat{\xi}$, the image of this set is easily seen to be given by
\begin{equation}\label{eq:imhatxi}
    \hat{\xi}\left(f(D)\setminus f\left(\frac{1}{2}D\right)\right)=\left\{\left(\frac{x}{2\lVert x\rVert^2},[x]\right)\in Bl_0(D)\mid \frac{1}{2}<\lVert x\rVert <1\right\}
\end{equation}
Note that the fact that the image is defined on $Bl_0(D)$ means that it is still subject to the set of equations which defined it (c.f. \eqref{eq:blowup2}). In particular the above restriction is seen to be a diffeomorphism of the annulus $D\setminus \frac{1}{2}D$ to itself, namely $x\mapsto (\frac{x}{2\lVert x\rVert})$. This map precisely defines the gluing scheme. Finally, due to the blow-up equations, the remaining components of the image are completely fixed in $\mathbb{CP}^{n-1}$

To see that this map also defines a diffeomorphism to $Bl_0(D)$, we go back \eqref{eq:blowup}. The only possibility for $\hat{\xi}$ to not be differentiable is that $x=0$. However, this would imply $|x_0|<0$, which is impossible by \eqref{eq:cpdisk}. Furthermore, we can see that the behavior at the singular point $0\in D$, which was the original purpose of the blow-up, is reproduced by this map, since $x_0=0$ is an allowed value which, by both \eqref{eq:cpdisk} and the equations that define the blow-up, lets $x$ range freely over all of $\mathbb{CP}^{n-1}$. Finally, we can see that $\text{Im}(\hat{\xi})$ is all of $Bl_0(D)$ by taking the norm of the first $n$ components of \eqref{eq:imhatxi} and again noting that by \eqref{eq:cpdisk} $|x_0|<2\lVert x\rVert$, thus arriving at
\begin{equation}
    \frac{|x_0| \lVert x\rVert}{2\lVert x\rVert^2}<1
\end{equation}
\end{proof}

\subsection{Blowups and magnetic monopoles}
Our end goal is to somehow apply the Ricci flow machinery to the manifold which is imposed by the cobordism conjecture. In principle, to be able to do so, we first need to write down an explicit form of the metric for the whole manifold. This is tackled in \cite{nCP2s}, and both the result, and the method are relevant for our purposes, since a connection is drawn between the successive blowup of points on a manifold and the connected sum of projective planes. This is further supported by the ADE classification of K3 singularities, and its interplay with blowups of said singularities. This last point will be left for a following section, and we will rather focus our attention here on the geometric side of the construction. We will leave out some details, which can be read on the original publication.

The starting point for the construction of the metric is the Gibbons-Hawking \textit{ansatz} for gravitational multi-instantons, first proposed in \cite{GibbonsHawking}. This is generalized to allow for an isometric circle action on the manifold, and reformulated in geometrical terms, as solutions of a certain set of partial differential equations. The central result of the paper, out of which all of the results are derived, is the following:
\begin{prop}\label{prop:cp}
Let $\xi>0$  and $u$ be smooth real-valued functions on an open set $U\subseteq \mathbb{R}^3$ which satisfy 
\begin{align}
    u_{xx}+u_{yy}+(e^u)_{zz}&=0 \\
    \xi_{xx}+\xi_{yy}+(\xi e^u)_{zz}&=0
\end{align}
Suppose further that the de Rham class of the closed 2-form
\begin{equation}
    \frac{1}{2\pi}\alpha\coloneqq \frac{1}{2\pi}\left(\xi_x dy\wedge dz+\xi_y dz\wedge dx+ (\xi e^u)_z dx\wedge dy\right)
\end{equation}
is contained in $H_{dr}^2(U,\mathbb{Z})$ (i.e., it is integral). Let $M\to U$ be a circle bundle such that its first Chern class is given by $\left[c_1(M)\right]_{\mathbb{R}}=\left[\frac{\alpha}{2\pi}\right]$, and let $\omega$ be a connection 1-form on the bundle whose curvature is $\alpha$. If $U$ is simply connected, both $M$ and $\omega$ are determined up to gauge equivlence. Then 
\begin{equation}
    g=e^u\xi(dx^2+dy^2)+\xi dz^2+\xi^{-1}\omega^2
\end{equation}
is a K\"ahler metric on $M$ whose scalar curvature vanishes. Conversely, any scalar-flat K\"ahler surface with an $\mathbb{S}^1$ action can be locally described by the above result.
\end{prop}
Note that the above metric reduces to the Gibbons-Hawking \textit{ansatz} if $u=0$. The space that this metric represents is also knwon as the Taub-NUT space. This is very relevant to our goal for two reasons. The first one is that the Taub-NUT space describes a KK monopole: a purely geometrical solution to Einstein's equations in $d+1$ dimensions, which from the perspective of the $d$ dimensions transverse to the circle bundle describes a magnetic monopole. Moreover, the Taub-NUT solution allows us to superpose $N$ of these monopoles. While we will not deal with Taub-NUT spaces, but rather with a space that is very similar to it, a lot of the properties will be shared between the two. Thus, it is useful to keep this solution in mind for the following.

The connection that all of this has with blow-ups (and thus with connected sums of copies of $\mathbb{CP}^2$) is given by the so-called Burns metric. This is a scalar flat K\"ahler metric on $Bl_0(\mathbb{C}^2)$, the blow-up of $\mathbb{C}^2$ at the origin. There are several coordinate representations of this metric, and we will come back to this in a later section. For now, it suffices to know that this metric fits into the description of proposition \ref{prop:cp} by setting 
\begin{align}
    u &=\log 2z \\
    \xi &=\frac{1}{2z}+F_m \label{eq:cfV}
\end{align}
where $m$ is some parameter which appears on the metric, and $F_m(x,y,z)$ is a positive function on $z>0$, and is defined on the complement of $(x,y,z)=(0,0,m/2)$. In particular, this makes sense when the base space is $\mathbb{H}^3\coloneqq \{(x,y,z)\in \mathbb{R}^3\mid z>0\}$.

If one defines a new variable $V=2z\xi$, and further defines $q=\sqrt{2z}$, it turns out that the second PDE from proposition \ref{prop:cp} can be recast as
\begin{equation}\label{eq:vdef}
    \Delta V=0,
\end{equation}
where $\Delta$ is the Laplace-Beltrami operator associated to the metric
\begin{equation}\label{eq:hdef}
    h=\frac{dx^2+dy^2+dq^2}{q^2}.
\end{equation}
This is just the usual metric for hyperbolic three-space, which we denote by $\mathcal{H}^3$. In terms of $V$ and $h$, the metric can be written as 
\begin{equation}
    g=q^2[Vh+V^{-1}\omega^2]
\end{equation}
Furthermore, we can use the hyperbolic metric distance to rewrite the function $\xi$ (and thus $V$) that appears in the metric. With the metric \eqref{eq:hdef} for $\mathbb{H}^3$, the hyperbolic distance from a generic point $(x,y,q)\in \mathbb{H}^3$ to $(0,0,q_0)\in \mathbb{H}^3$ can be written as 
\begin{equation}
    \rho=\cosh^{-1}\left(\frac{x^2+y^2+q^2}{2qq_0}\right)
\end{equation}
After a couple of algebraic manipulations, we can recast $V$ as 
\begin{equation}
    V=1+\frac{1}{e^{2\rho}-1}
\end{equation}
But this last summand is precisely the Green's function $G$ for the Laplace-Beltrami operator relative to the hyperbolic volume form, and with normalization given by $\Delta G=-2\pi\delta$. Consequently, the Burns metric can also be thought as describing a magnetic monopole in hyperbolic 3-space.

Up to now, we have seen that the metric for a blowup of $\mathbb{C}^2$ at the origin can be recast (or, rather, understood) as a kind of generalized Gibbons-Hawking metric, which looks eerily similar to a Taub-NUT geometry. Now, the next step, as is commonplace precisely with the Taub-NUT metric, is to allow for multiple monopoles for our metric. This can be done as follows: let $\{p_j=(x_j,y_j,z_j)\}_{j=1}^n\subset \mathbb{H}^3$ be a collection of points on the upper half-space, and let 
\begin{equation}
    G_j=\frac{1}{e^{2\rho_j}-1}
\end{equation}
be the Green's functions associated to each of the points in the sense of the above paragraph. Define
\begin{equation}\label{eq:defV}
    V\coloneqq 1+\sum_{i=1}^nG_j                      
\end{equation}
Then, we have that $\xi=\frac{V}{2z}$ is a positive solution for 
\begin{equation}
    \xi_{xx}+\xi_{yy}+(2z\xi)_{zz}=-2\pi\sum_{i=1}^n\delta_{p_i}.
\end{equation}
This defines an integral cohomology class as 
\begin{equation}
    \frac{1}{2\pi}\alpha= \frac{1}{2\pi}\left(\xi_x dy\wedge dz+\xi_y dz\wedge dx+ (\xi e^u)_z dx\wedge dy\right)\in H^2(\mathbb{H}^3\setminus \{p_i\},\mathbb{Z})
\end{equation}
Finally, let $X\to\mathbb{H}^3\setminus\{p_i\}_{i=1}^n$ be the circle bundle whose first Chern class is $\frac{1}{2\pi}[\alpha]$. In turn, this means that $X$ has a connection 1-form $\omega$ whose curvature is given (up to gauge equivalence) by $\alpha$. Then, define a Riemannian metric on $X$ by 
\begin{equation}\label{eq:lametrica}
    g=q^2[Vh+V^{-1}\omega^2]
\end{equation}
By proposition \ref{prop:cp}, this is a K\"ahler metric with vanishing scalar curvature. However, it is still not quite the metric that we are looking for. For one, it is not defined over the "monopole locations", namely the $\{p_i\}_{i=1}^n$. However, it can be smoothly extended to them by attaching points at the monopole locations. This is done in detail in \cite{nCP2s}, and contains very technical details which are not very relevant for our purposes. The only part of this discussion which we will use is a local expression for the metric around the monopole locations. This will be discussed in a later section.

In summary, we are able to construct the K\"ahler metric \eqref{eq:lametrica} which, by the above discussion, can be interpreted as describing the blow-up of $\mathbb{C}^2$ at $n$ points, which can be furhter viewed as the locations of magnetic monopoles in hyperbolic 3-space. Since blow-ups and connected sums with copies of $\mathbb{CP}^2$ are diffeomorphic, a consequence of this construction is that the conformal class of the metric \eqref{eq:lametrica} represents a metric on $n\mathbb{CP}^2\coloneqq \mathbb{CP}^2\#\stackrel{n)}{\ldots}\#\mathbb{CP}^2$. In fact, a global representative for the conformal metric on $n\mathbb{CP}^2$ is provided by 
\begin{equation}\label{eq:lametrica2}
    g=\text{sech}^{2}(\rho)[Vh+V^{-1}\omega^2]
\end{equation}
In the following, we will use this very metric, as well as some of the intermediate steps that led to it, to try to apply the Ricci flow equations to $n\mathbb{CP}^2$. In this way, we hope to use the flow equations to obtain information about backgrounds which are allowed by the cobordism conjecture.

\subsection{The Burns metric}

Now that we have some geometrical information about the manifold $n\mathbb{CP}^2$, the next step is to set up some kind of flow equations, or at least to approximate the manifold's behavior under them. Because of the above discussion, the first step that we will take in this endeavor is to look closely at the Burn's metric. As we have mentioned before, this is a zero scalar curvature metric on $Bl_0(\mathbb{C}^2)$ (we are using the notation from the previous sections for the blowup). Because we have seen that blowups at a point in a manifold are diffeomorphic to connected sums of said manifold with $\mathbb{CP}^2$, we hope to extract some information out of this.

\subsubsection {First form of the Burns metric}

There is more than one way to represent the Burn's metric. Here we will only make use of two of them. In the first representation, the K\"ahler nature of the metric is manifest, as it is given by the following K\"ahler potential defined over $\mathbb{C}^2\setminus \{0\}$:
\begin{equation}
    K(z,\overline{z})=\frac{1}{2}\left(\lVert z \rVert^2+m\log\lVert z\rVert^2\right)
\end{equation}
where $m>0$ is a positive real constant, and $\lVert z\rVert^2=z_1\overline{z}_1+z_2\overline{z}_2$. Thus, the metric can be obtained as 
\begin{equation}
    g_{i\overline{j}}=\partial_i\overline{\partial}_{\overline{j}}K(z,\overline{z})
\end{equation}
Going through the computations, this yields 
\begin{equation}
    g_{1\overline{1}}=\frac{1}{2}\left(1+\frac{m\mid z_2\mid^2}{\left(\lVert z\rVert^2\right)^2}\right);\hspace{1.5cm} g_{2\overline{1}}=-\frac{m}{2}\left(\frac{z_1}{\left(\lVert z\rVert^2\right)^2}\overline{z}_2\right)
\end{equation}
with all other entries of the metric given by the fact that it is Hermitean (so we need only exchange $1\Leftrightarrow 2$):
\begin{equation}
    g=\frac{1}{2}\begin{pmatrix}
    \left(1+\frac{m\mid z_2\mid^2}{\left(\lVert z\rVert^2\right)^2}\right) & -m\left(\frac{z_2}{\left(\lVert z\rVert^2\right)^2}\overline{z}_1\right) \\
    -m\left(\frac{z_1}{\left(\lVert z\rVert^2\right)^2}\overline{z}_2\right) & \left(1+\frac{m\mid z_1\mid^2}{\left(\lVert z\rVert^2\right)^2}\right)
    \end{pmatrix}
\end{equation}

For Kähler manifolds, there is a quick way of computing the Ricci tensor, namely by 
\begin{equation}\label{eq:rickaeh}
    R_{i\overline{j}}=-\partial_i\overline{\partial}_{\overline{j}}\log\det g.
\end{equation}
A quieck computation reveals that 
\begin{equation}
    4\det g=1+\frac{m}{\lVert z\rVert^2}.
\end{equation}
Note that setting $m=0$ is consistent with recovering the usual flat metric on $\mathbb{C}^2$. Since we are taking derivatives of logs, we may also drop the overall factor of $\frac{1}{4}$, thus arriving at 
\begin{equation}
R_{i\overline{j}}=-\partial_i\overline{\partial}_{\overline{j}}\log\left(1+\frac{m}{\lVert z\rVert^2}\right)
\end{equation}
The resulting tensor is given by 
\begin{equation}
    R_{2\overline{1}}=\frac{mz_1}{\left[\left(\lVert z\rVert^2\right)^2+m\lVert z\rVert^2\right]^2}\left(2\lVert z\rVert^2\overline{z}_2+m\overline{z}_2\right);\hspace{0.5cm} R_{1\overline{1}}=-\frac{m}{\left(\lVert z\rVert^2\right)^2+m\lVert z\rVert^2}+\frac{\lvert z_1\rvert^2\left(2\lVert z\rVert^2+m\right)}{\left[\left(\lVert z\rVert^2\right)^2+m\lVert z\rVert^2 \right]^2}
\end{equation}
and the other components are again obtained by the exchange of indices. As one can appreciate, the resulting Ricci flow equations become rather complicated, and far form the case for an Einstein manifolds. It would have been our hope to be able to get some sort of flow equation for the parameter $m$, but it seems that our hopes are in vain. Even if we tried to do so out of the Ricci flow equation for the scalar curvatrue, we are met with the fact that this metric is, by construction, of zero scalar curvature. Hence, if we were to set $m=m(s)$ for some flow parameter $s$, our equation would render $m$ constant along the flow.

Given the explicit formula for the Ricci tensor \eqref{eq:rickaeh}, another idea that one could have is to try to take the limit where $m\ll\lVert z\rVert$, in order to get rid of the logarithm. Unfortunately, this has led me nowhere. The only interesting point in the "moduli space", where the Ricci tensor turns out to vanish, requires $\lVert z\rVert^2\to 0$, which is inconsistent with the limit. On the other hand, we have noted that (even at the level of the Kähler potential itself) the limit $m\to 0$ corresponds to flat space.

\subsection{Second form of the Burn's metric}
There is yet another way to represent the Burn's metric, in which its K\"ahler nature is rather implicit. However, the upside of it is that we can obtain a (real) coordinate representation of it. Hence, we can use a symbolic computation program, such as Mathematica, to more easily compute all the quantities that we need. Moreover, this representation is instructive when we compare it with local representations of the metric \eqref{eq:lametrica2}.

Let $\{\sigma_1,\sigma_2,\sigma_3\}$ be the coframe of $SU(2)$ left-invariant 1-forms, normalized such that $\de\sigma_1+2\sigma_2\wedge\sigma_3=0$ (and cyclically). Since $SU(2)\simeq \mathbb{S}^3$, these forms also serve as an orthonormal coframe on the three-sphere. Then, the Burns metric can be given the following coordinate representation:
\begin{equation}\label{eq:burnscoord}
    g_B=dr^2+r^2\left(\sigma_1^2+\sigma_2^2+\sigma_3^2\right)+\sigma_1^2+\sigma_2^2
\end{equation}
As a comment, note that in this representation the usual Euclidean metric is given by $g=dr^2+r^2(\sigma_1^2+\sigma_2^2+\sigma_3^2)$. As we have previously discussed, the metric \eqref{eq:burnscoord} describes the blow-up of $\mathbb{C}^2$ at the origin. 

In order to do any kind of calculation with the Burns metric, we need to introduce an explicit coordinate representation of the left-invariant $SU(2)$ one-forms. We do so in terms of the Euler angles, which we denote by $\{\theta,\varphi,\psi \}$. The one-forms then take the following explicit form:
\begin{equation}\label{eq:su2euler}
\begin{cases}
2\sigma_1&=\sin(\psi)d\theta-\cos(\psi)\sin(\theta)d\varphi \\
2\sigma_1&=\cos(\psi)d\theta+\sin(\psi)\sin(\theta)d\varphi \\
2\sigma_3&=d\psi+\cos(\theta)d\varphi
\end{cases}
\end{equation}
With these explicit coordinate representations, we are able to compute both the Ricci scalar and the Ricci tensor. As should be expected, we have that $R_B=0$ identically. We will come back to this point in a later section

\subsection{Getting a feel for the flow}
Seeing that the treatment of the metric for $C\mathbb{P}^2\#\stackrel{n)}{\ldots}\#C\mathbb{P}^2$ explicitly and analitically is rather complicated, we turn to approximating the behavior of the manifold under the Ricci flow. To this end, we will use less information than the required to solve the full flow equations. Namely, we turn to a result about the metric \eqref{eq:lametrica2} which is stated in \cite{nCP2s}. While the original metric \eqref{eq:lametrica} was constructed as a scalar flat metric, the actual representative for the metric on $n\mathbb{CP}^2$ has a different choice of conformal factor, and thus the scalar curvature need not vanish. In fact, it does not, and it turns out to be
\begin{equation}\label{eq:Rncp2}
    R=\frac{12}{V}.
\end{equation}
Where $V$ is given by \eqref{eq:defV}. Since $\rho_j$ is the distance from a point in $\mathcal{H}^3$ to one of the monopole locations, it only vanishes at these singular loci. As a consequence of this, at these loci, $V\to \infty$. Conversely, the scalar curvature vanishes at (and only at) these singular points. 

In order to use this information to study the behavior of the Ricci flow on this manifold, we need an additional result, coming from the general theory of the flow equations. In particular, as a consequence of the maximum principle, we have that 
\begin{theorem}
Let $g(t)$ be a Ricci flow on a closed $d$-dimensional manifold $X$, for $t\in[0,T]$. If $R\geqslant \alpha \in \mathbb{R}$ at time $t=0$, then for all times $t\in [0,T]$,
\begin{equation}
    R\geqslant\frac{\alpha}{1-\left(\frac{2\alpha}{d}\right)t}
\end{equation}
\end{theorem}
What is  important about this result is that, whenever the lower bound is strictly positive, the scalar curvature blows up at some finite flow time. Going back to \eqref{eq:Rncp2}, and recalling the explicit formula for $V$, we see that $R$ is precisely a non-negative function over $n\mathbb{CP}^2$. Moreover, the only points at which $R$ tends to 0 are precisely the singular loci which define the location of the magnetic monopoles. This is enough information to know that, at every point but the singularities, there is a lower bound on the scalar curvature which goes to infinity at finite flow time. Hence, we can guarantee that the curvature blows up and the manifold shrinks infinitely (but at finite time) at all points but at the singular loci. 

 \subsubsection{The behavior of the flow at the monopole locations}
 
 A failure of the above argument is that the scalar curvature \eqref{eq:Rncp2} vanishes at the location of the monopoles. To be more precise, $R$ is not technically defined at the $\{p_i\}$. However, we have already said that it can be smoothly extended to the whole manifold by essentially attaching points at the monopole locations. Then, after a series of approximations of the metric around said locations, one finds that it can indeed be smoothly extended to them. Because of the smoothness of this extension, it follows that the Ricci scalar also extends smoothly, and hence by continuity it must vanish at the monopole locations, since 
 \begin{equation}\label{eq:Rloci}
     \lim_{p\to p_i}\frac{1}{V}=0.
 \end{equation}
This, of course, renders the argument from the previous section invalid, and hence we cannot use it to justify that the manifold shrinks at every point. One could, in principle, follow the approximations which lead to the extension of the metric, and try to see if anything can be done to set up a flow around the singularities. This is what we will attempt now.

As we have said, the full manifold is obtained by an extension process which, at the monopole locations in particular, begins by attaching points at the missing $\{p_i\}$. This can be done in such a way that the metric that we had previously built extends smoothly to the the full manifold. To this end, we begin by introducing exponential polar coordinates on $\mathbb{H}^3$ near one of these points. The metric then becomes
\begin{equation}
    g=q^2[V(d\rho^2+\sinh^2(\rho) g_{\mathbb{S}^2})+V^{-1}\omega^2]
\end{equation}
where $V=\frac{1}{2\rho}+F$ (cf. \eqref{eq:cfV}), for $F$ some smooth function on a neighbourhood of the origin, and $g_{\mathbb{S}^2}$ denotes the standard metric on the unit 2-sphere. In terms of the $SU(2)$ left-invariant one-forms that we introduced earlier on, this can be written as $g_{\mathbb{S}^2}=\sigma^2_1+\sigma_2^2$. For small values fo the radial coordinate $\rho$, we can use the polar coordinates on the base to identify it with $\mathbb{R}^+\times \mathbb{S}^2$, and the circle bundle with $\mathbb{R}^+\times \mathbb{S}^3$. In turn, the bundle map $\pi:X\to \mathbb{H}^3\setminus\{p_i\}$ is identified with 
\begin{align*}
    p: \mathbb{R}^+\times \mathbb{S}^3&\to\mathbb{R}^+\times \mathbb{S}^2 \\
    (r, x)\hspace{0.2cm}&\mapsto \left(\frac{r^2}{2}, \mu(x)\right).
\end{align*}
Here, the map $r\mapsto r^2/2$ is chosen as such for consistency with the previously constructed bundle map, $pi:X\to \mathbb{H}^3\setminus\{p_i\}$, and $\mu:\mathbb{S}^3\to \mathbb{S}^2$ is the usual Hopf fibration. Because of this identification with unit spheres, it is now useful to fully introduce the left-invariant $SU(2)$ coframe $\{\sigma_1, \sigma_2, \sigma_3\}$ on $\mathbb{S}^3$, again normalized so that $\de\sigma_1+2\sigma_2\wedge\sigma_3=0$. Because of this relation, and the fact that the curvature of the connection form is defined over the base space, we have the (gauge) freedom of choosing $\omega=\sigma_3+p^*\theta$, with $\theta$ some one-form defined on the base space. With all these choices, and the change of variable $\rho\to r^2/2$, the metric turns out to be 
\begin{equation}\label{eq:oofmetric}
    g=q^2\left[(1+r^2F)dr^2+r^2(1+\sinh^2\left(r^2/2\right)r^2F)(\sigma_1^2+\sigma^2_2)+r^2(1+r^2F)^{-1}(\sigma_3+p^*\theta)^2\right]
\end{equation}
Identifying $\mathbb{R}^+\times\mathbb{S}^3$ with $\mathbb{R}^4\setminus\{0\}$, the map $p$ extends to 
\begin{align*}
    p:\mathbb{C}^2&\to\mathbb{R}^3 \\
    (z_1,z_2)&\mapsto \left(\frac{|z_1|^2-|z_2|^2}{2}, \text{Re}(z_1\overline{z}_2), \text{Im}(z_1\overline{z_2})\right)
\end{align*}
In particular, this identification allows us to identify the metric at the origin as precisely the Euclidean metric, expressed as $g=dr^2+r^2(\sigma_1^2+\sigma_2^2+\sigma_3^2)$. In any case, the fact that the original metric agrees with the Euclidean one at the origin in no way allows us to conclude that it is locally flat. As a final remark, note that one still needs to implement the appropriate conformal factor $\text{sech}^2(\rho)$ in front of the metric for it to be the actual metric for $n\mathbb{CP}^2$.

After this whole discussion, one might be tempted to try to set up the Ricci flow equations for \eqref{eq:oofmetric}, or an approximation at the origin thereof. This is rather complicated. Most importantly, while we could explicitly compute the function $F$, the one-form $\theta$ is arbitrary. This prevents us from giving an explicit form of the metric, and thus from making explicit computations with it. Because in the context of \cite{nCP2s}, the explicit form of $\theta$ is irrelevant, no further properties are given for it. If one were to drastically simplify the metric \eqref{eq:oofmetric}, and assume both that $F$ is constant, and that $p^*\theta=0$, then the metric that we would be dealing with would be very similar to the Euclidean metric, or the Burns metric, with different prefactors, which have simple dependences on the coordinate $r$. With this very drastic oversimplification, we obtain that $R=2+8F$, which, by virtue of $F$ being positive, is itself positive. This tells us that this approximation is not terribly useful, since it contradicts the argument given around \eqref{eq:Rloci}, and thus renders the approximation invalid. We could try to relax the conditions, and take $p^*\theta$ to have constant coefficients in the $\{dr, \sigma_1,\sigma_2\}$ basis for the base space. This indeed allows us to force $R$ to vanish, depending on the value of the coefficients. However, we are still approximating $F$ for a constant, when it should be a function of the coordinates in the base. Hence, the approximation seems too drastic to be meaningful.

\subsection{Addendum: Why the behavior at the singularities may not matter}

In the previous sections, we have tried to argue that the Ricci flow equations make the whole $n\mathbb{CP}^2$ manifold, endowed with the metric \eqref{eq:lametrica2} shrink. While we were able to check that this does indeed happen almost everywhere, we found some trouble on the set of points $\{p_i\}\subset n\mathbb{CP}^2$, where the monopoles are located, and the original space was blown up. In this section, we argue that the behaviour of the manifold at these discrete points might not be all that important.

What we know for sure is that the Ricci flow makes $n\mathbb{CP}^2\setminus\{p_i\}$ shrink indefinitely until, at the limit where the flow ends, it shrinks to a point. Notice that we have used the word \textit{limit} here, since turning an $n$-dimensional manifold to a point is a procedure that does great violence to almost any conceivable topological property, and thus, cannot be a part of a smooth deformation. However, the interpretation still stands. If we take the manifold to shrink to a point, then taking so much care about the exact behavior at the monopole locations might be pointless, since we are already heavily modifying the manifold. 

Another argument that can support this claim is the following. Recall that, in order to make sense of the idea of a manifold converging along the flow to some other manifold, we introduced in section \ref{sect:necks} the notion of smooth, pointed convergence. In this context, we could pick precisely the set $\{p_i\}$ of points along the flow. Indeed, since the curvature increases much faster away from the monopole locations than around them, we expect necks to form, and to then pinch off. In this sense, the Ricci flow would perform a disconnected sum decomposition of $n\mathbb{CP}^2$, in such a way that it would leave a point $p_i$ in each of the remaining components. Then, under the notion of pointed convergence, we  conclude that the manifold would converge to a set of points, which are precisely the monopole locations.

\section{Distances along the flow}
The Swampland distance conjecture \cite{Ooguri:2006in,Palti:2019pca} substantially asserts that large distances in the effective field theory moduli space should be accompanied by asymptotically massless towers of states, when quantum gravity is taken into account. This intuition was further extended to the case of Ricci flow in \cite{Kehagias:2019akr}, in which metric components themselves were regarded as moduli of the theory. Exploiting the technical discussion presented \cite{GilMedrano1991THERM}, the distance along a family of metrics $g\left(t\right)$ on a Riemannian manifold $\mathcal{M}$, where $t$ is a real flow parameter, can be computed as
\begin{equation}
    \Delta\left(t_1,t_2\right)=c\int_{t_1}^{t_2}\left(\frac{1}{\mathcal{V}_{\mathcal{M}}}\int_{\mathcal{M}}\sqrt{g}g^{\mu\nu}g^{\alpha\beta}\frac{\de g_{\mu\alpha}}{\de t}\frac{\de g_{\nu\beta}}{\de t}\right)^{\frac{1}{2}}\de t\ ,
\end{equation}
where $c$ is a real constant and $\mathcal{V}_{\mathcal{M}}$ is the volume of $\mathcal{M}$. When the family is taken to follow Ricci flow, the above reduces to:
\begin{equation}\label{eq:riccidistance}
    \Delta\left(t_1,t_2\right)=2c\int_{t_1}^{t_2}\left(\frac{1}{\mathcal{V}_{\mathcal{M}}}\int_{\mathcal{M}}\sqrt{g} R_{\mu\nu}R^{\mu\nu}\right)^{\frac{1}{2}}\de t\ .
\end{equation}
For a $D$-dimensional Einstein manifold, we get:
\begin{equation}
    \Delta\left(t_1,t_2\right)\sim\log{\frac{R\left(t_1\right)}{R\left(t_2\right)}}\ .
\end{equation}
Therefore, it is clear that both flow singularities and flat spacetime fixed points lie at infinite distance and should be accompanied by a tower of asymptotically massless states. For a more general manifold, we can decompose the Ricci scalar as
\begin{equation}
    R_{\mu\nu}=\frac{R}{D}g_{\mu\nu}+T_{\mu\nu}\ ,
\end{equation}
where $T_{\mu\nu}$ is its traceless part. Since
\begin{equation}
    R_{\mu\nu}R^{\mu\nu}=\frac{R^2}{D}+T_{\mu\nu}T^{\mu\nu}
\end{equation}
we obtain
\begin{equation}
    \Delta\left(t_1,t_2\right)=2c\int_{t_1}^{t_2}\biggl(D^{-1}\braket{R^2}+\braket{T_{\mu\nu}T^{\mu\nu}}\biggr)^{\frac{1}{2}}\de t\ ,
\end{equation}
in which $\braket{A}$ referes to the average value of $A$ on $\mathcal{M}$. Let's assume that the flow encounters a singularity at $t=t_{s}$, so that $R\to\infty$ for $t\to t_s$. In the broad set of situations in which it implies $\braket{R^2}\to\infty$, we can expect that flow singularities sit at an infinite distance in the moduli space. This is clearly the case for a positive-curvature manifold shrinking to a point, as a $\mathbb{CP}^2$ term in a disjoint union. It would be interesting to assess this in a more general setting and investigate what singularity resolutions imply from the perspective of the distance along the flow. The fact that they connect cobordant manifolds gives hope that one could somehow assign a meaningful moduli space length to the deformation connecting the manifold which is about to pinch and the resolved one, even though a topology change is involved. We expect that rephrasing \eqref{eq:riccidistance} in a more general form, in which the deformation parametrised by $t$ can generally connect cobordant manifolds, could shed some light on the issue at hand. The $t$-integral, in that sense, would have to be regarded as an integral over the extra coordinate of the manifold of which the initial and the resolved one are boundaries. In particular, one would like to assess whether singularity resolutions can avoid distance divergencies. This would be consistent with the general expectation that quantum gravity processes can naturally induce topology changes, implying that space-time manifolds with different topological properties should belong to the same moduli space. Further discussions will be required, but such an investigation goes beyond the scope of the current work.

\section{The ADE singularities}
One of the most striking features of $K3$ surfaces, which is also very relevant for physics applications, is that all of their singularities are classified in terms of the ADE Lie groups. However, when one resolves these singularities by blowing them up, all the information regarding the ADE classification is lost. From the connected sum perspective, the manifold $K3\#\mathbb{CP}^2$, for example, contains no information about the singularity which was replaced by a smooth manifold. In this section, we will try to elucidate how to recover the ADE information from the blow-up of the singularities. The idea is that the cobordism conjecture states that $K3$ is a valid background only when it is accompanied by a series of blow-ups of its singularities. This is sometimes solved by adding a series of defects on the $K3$, sourcing fluxes. Then, it could be that we can reconstruct, or find some restrictions on, the gauge groups associated to said fluxes.

\subsection{A lightning fast overview of complex geometry}
In this section we review/introduce some concepts in algebraic and complex geometry which are needed to shape the discussion. They are not, however, necessary to understand the general idea. We begin with the concept of a holomorphic line bundle.

\subsubsection{Holomorphic line bundles and the Picard group.}

In general terms, a line bundle is simply  a vector bundle or rank $1$, so the fibers at any point are lines. Note the difference on the nature of the bundle when its dimension is complex rather than real. We now give the definition of a holomorphic vector bundle
\begin{mydef}
Let $X$ be a complex manifold. A holomorphic vector bundle of rank $r$ on $X$ is a complex manifold $E$ together with a holomorphic projection map $\pi:E\to X$, and such that each fibre $E(x)\coloneqq \pi^{-1}(x)$ is an $r$-dimensional complex vector space, which further satisfies the following condition: There exists an open covering $X=\bigcup_{i}U_i$ and biholomorphic\footnote{A biholomorphic function is a function $f:U\to \mathbb{C}^n$ which is both holomorphic and injective.} maps $\psi_i:\pi^{-1}(U_i)\stackrel{\sim}{\longrightarrow} U_i\times \mathbb{C}^r$ commuting with the projections from the bundle to $U_i$ such that the induced map $\pi^{-1}(x)\stackrel{\sim}{\longrightarrow}\mathbb{C}^r$ is complex linear.
\end{mydef}
The above definition os basically that of a regular vector bundle, but we further require that the so-called \textit{local trivializations} (the $\psi_i$ above) are holomorphic. Now, we can just define a \textit{holomorphic line bundle} as a holomorphic vector bundle of rank $r=1$. One last comment on this definition is that the induced transition functions on the vector bundle 
\begin{equation}
    \psi_{ij}(x)\coloneqq\left(\psi_i\circ\psi_{j}^{-1}\right)(x, \ ):\mathbb{C}^r\to\mathbb{C}^r
\end{equation}
are complex linear for all $x\in U_i\cap U_j$. In the same way as for differentiable real or complex bundles, a holomorphic rank $r$ vector bundle is determined by the set of holomorphic cocycles $\{U_i,\psi_{ij}:U_i\to GL(r,\mathbb{C})\}$. This will be important in the following, when we introduce the Picard and the Néron-Severi (NS) groups.

We now turn to a very specific (and unique) type of line bundle. Over an $n$-dimensional complex projective space $\mathbb{CP}^n$, there is (up to isomorphism) only one holomorphic line bundle. It is called the \textit{tautological line bundle}, and denoted $\mathcal{O}(1)$. Its dual is defined as follows:
\begin{prop}
The set $\mathcal{O}(-1)\subseteq \mathbb{CP}^n\times \mathbb{C}^{n+1}$, consisting of all pairs $(l, z)\in \mathbb{CP}^n\times \mathbb{C}^{n+1}$ such that $z\in l$ forms in a natural way a holomorphic line bundle over $\mathbb{CP}^n$.
\end{prop}
The proof of this proposition is not terribly complicated, and hence we will skip it. The basic idea is to just take an element in $\mathbb{CP}^n$, which is a complex line, and define the fibre over that point as the line itself. Rigorously, it can be shown that the fibre $\mathcal{O}(-1)\to \mathbb{CP}^n$ over $l \in \mathbb{CP}^n$ is naturally isomorphic to $l\subset \mathbb{C}^{n+1}$. The above construction gives rise to the following definition
\begin{mydef}
The line bundle $\mathcal{O}(1)$ is defined as the dual of $\mathcal{O}(-1)$, namely $\mathcal{O}(-1)^*$. This can be extended. For $k\in \mathbb{Z}^+$, let $\mathcal{O}(k)$ be the line bundle $\mathbb{O}(1)\otimes\stackrel{k)}{\ldots}\otimes \mathcal{O}(1)$. One defines $\mathcal{O}(-k)$ in an analogous way. It is standard notation to write $E(k)\coloneqq E\otimes \mathcal{O}(k)$ for any vector bundle $E\to \mathbb{CP}^n$.
\end{mydef}
Finally, if one defines $\mathcal{O}(0)\coloneqq \mathbb{CP}^n\times \mathbb{C}$ (i.e. the trivial line bundle), the tensor product and the dual endow the set of all isomorphism classes of holomorphic line bundles over a complex manifold $X$ with the structure of an Abelian group. Being that this groups is infinite cyclic, it must be isomorphic to $\mathbb{Z}$. This is the Picard group of $X$, $\text{Pic}(X)$.

There is a consequence of the above definition which will help us build up to our end goal. It is the fact that $\text{Pic}(X)\simeq H^{1}(X, \mathcal{O}_X^*)$. Here, $\mathcal{O}_X^*$ denotes the subsheaf of nonvanishing holomorphic functions on $X$. The precise definition of a sheaf is not necessary here. To understand the general argument it is enough to thing about it in terms of homology with coefficients which are nonvanishing holomorphic functions. The proof for the previous statement, which we skip here, relies on two facts. The first one is that line bundles are described in terms of cocycles, given by the transition functions. These are in fact not any type of cocycles, but fall under the classification of Čech cocycles. The association is clear with the definition of the Čech cohomology groups. We will skip this definition for reasons that will be clear in a moment. Because of this definition in terms of cocycles, one has that $\text{Pic}(X)\simeq \check{H}^{1}(X, \mathcal{O}_X^*)$, where $\check{H}$ denotes \v{C}ech cohomology. The second fact that this proof relies on, is that for $p=1$, the homomorphism $\check{H}^p(X,\mathcal{F})\to H^{p}(X, \mathcal{F})$ is an isomorphism. Therefore, in the end, we have that $\text{Pic}(X)\simeq H^{1}(X, \mathcal{O}_X^*)$.

\subsubsection{The Néron-Severi group.}
We will now introduce one of the last concepts which are needed to start discussing the blowup of singularities on $K3$ surfaces. It builds up from the previous section in the sense that it is a construction which is defined in terms of the Picard group of a complex manifold. This is the so-called Néron-Severi group.

We beign by the fact that there is a short exact sequence of sheaves, called the \textit{exponential} sequence (for reasons that will be clear in a moment). It is the following
\begin{equation}\label{eq:expseq}
    0\to \mathbb{Z}\to\mathcal{O}_X\to\mathcal{O}_X^*\to 0.
\end{equation}
Again, $\mathcal{O}_X$ denotes the sheaf of of holomorphic functions on $X$, and $\mathcal{O}_X^*$ the subsheaf of nonvanishing holomorphic functions on $X$. We may check for exactness as follows. The first map, $\mathbb{Z}\to \mathcal{O}_X$ is an inclusion, so it is trivially injective. The map in the middle is the one that gives the name to the sequence. It is given by 
\begin{align*}
    \exp:\mathcal{O}_X&\to \mathcal{O}_X^* \\
    f &\mapsto e^{2\pi i f}.
\end{align*}
Thus, it is clear that $\text{im}(\mathbb{Z}\to\mathcal{O}_X)=\ker(\mathcal{O}_X\to\mathcal{O}_X^*)$. Finally, surjectivity of $\mathcal{O}_X\to\mathcal{O}_X^*$ is granted by the (local) existence of the complex logarithm.

The exponential sequence gives rise in the usual sense to a long exact sequence in cohomology. We look in particular at
\begin{equation}
    \ldots \to H^1(X,\mathbb{Z})\to H^1(X,\mathcal{O}_X)\to H^1(X,\mathcal{O}^*_X)\to H^2(X,\mathbb{Z})\to \ldots
\end{equation}
Inserting the Picard group explicitly, we have 
\begin{equation}\label{eq:exactseq}
    \ldots \to H^1(X,\mathbb{Z})\to H^1(X,\mathcal{O}_X)\to \text{Pic}(x)\to H^2(X,\mathbb{Z})\to \ldots
\end{equation}
The last map above is known, it is precisely the first Chern class $c_1:\text{Pic}(X)\to H^2(X,\mathbb{Z})$. A map which is related to this will give us the definition of the Néron-Severi group.

The canonical boundary map in the long exact sequence above, arising from the exponential map, $\text{Pic}(X)\to H^2(X,\mathbb{Z})$ can be composed with the map $H^2(X,\mathbb{Z})\to H^2(X,\mathbb{C})$, which is simply induced by the inclusion $\mathbb{Z}\hookrightarrow\mathbb{C}$. On the other hand, a standard result of Hodge theory is that, for compact K\"ahler manifolds, one has the following decomposition
\begin{equation}
    H^n(X,\mathbb{C})=\bigoplus_{p+q=n}H^{p,q}(X)
\end{equation}
In particular, for $n=2$,
\begin{equation}\label{eq:hodgedec}
    H^2(X,\mathbb{C})=H^{2,0}(X)\oplus H^{1,1}(X)\oplus H^{0,2}(X)
\end{equation}
Before going further, note that due to complex conjugation on $H^n(X,\mathbb{C})$, we have that $\overline{H^{p,q}}(X)=H^{q,p}(X)$. This, combined with the fact that $\mathcal{O}_X$ is the sheaf of \textit{holomorphic} functions on the manifold, yields $H^2(X,\mathcal{O}_X)=H^{0,2}(X)$. 

With this information, we know that the long exact sequence given by the exponential sequence \eqref{eq:expseq} shows that the composition
\begin{equation}
    \text{Pic}(X)\to H^{2}(X,\mathbb{Z})\to H^2(X,\mathbb{C})\to H^2(X,\mathcal{O}_X)\left(=H^{0,2}(X)\right)
\end{equation}
is trivial, precisely because it fits into an exact sequence. Moreover, we have that
\begin{equation*}
    \text{Im}\left(\text{Pic}(X)\to H^2(X,\mathbb{C})\right)\subseteq \text{Im}\left(H^2(X,\mathbb{Z})\to H^2(X,\mathbb{C})\right)
\end{equation*}
We also have that $\text{Im}\left(\text{Pic}(X)\to H^2(X,\mathbb{C})\right)\subseteq \ker(H^2(X,\mathbb{C})\to H^{0,2}(X))$, where the last map is the usual projection given by the Hodge decomposition. Consider now $H^2(X,\mathbb{R})\subseteq H^2(X,\mathbb{C})$. It is invariant under complex conjugation, and further contains $\text{Im}(H^2(X,\mathbb{Z})\to H^2(X,\mathbb{C}))$, which must therefore also be invariant under complex conjugation, and hence fall into $H^{1,1}(X)$ in the Hodge decomposition. As a conclusion of all this, one finds that the image $\text{Pic}(X)\to H^2(X, \mathbb{C})$ is contained in the following homology group
\begin{equation}
    H^{1,1}(X,\mathbb{Z})\coloneqq \left(\text{Im}(H^2(X,\mathbb{Z})\to H^2(X,\mathbb{C}))\right)\cap H^{1,1}(X).
\end{equation}
The above definition is sometimes abbreviated as $ H^{1,1}(X,\mathbb{Z})\coloneqq H^2(X,\mathbb{Z})\cap H^{1,1}(X)$. It is a standard result (Lefschetz theorem on $(1,1)$ classes), that for $X$ a compact K\"ahler manifold, the map $\text{Pic}(X)\to H^{1,1}(X,\mathbb{Z})$ is surjective.

To conclude this section, we look at a slightly different map from the Picard group of a complex manifold. The image of the map $\text{Pic}(X)\to H^2(X,\mathbb{R})$ is called the Néron-Severi group of the manifold $X$, denoted $NS(X)$. It spans a finite dimensional real vector space $NS(X)\subset H^2(X,\mathbb{R})\cap H^{1,1}(X)$. Because of the surjectivity of $\text{Pic}(X)\to H^{1,1}(X,\mathbb{Z})$, it turns out that the natural inclusion $NS(X)\subseteq H^{1,1}(X,\mathbb{Z})$ is actually an equality. Finally, the rank of $NS(X)$ is denoted by $\rho(X)$ and called the \textit{Picard number}. 

In the following, we will not work with general complex manifolds, but rather with very specific and well behaved ones. In particular, we will deal with Calabi-Yau manifolds, for which we have that $H^1(X,\mathcal{O}_X)=0$. This can in fact be read off from the Hodge diamond of any Calabi-Yau $n$-fold. Because of the exact sequence \eqref{eq:exactseq}, this means that $c_1:\text{Pic}(X)\to H^2(X,Z)$ is injective. This in turn implies that 
\begin{equation*}
    \text{NS}(X)=\text{Pic}(X)
\end{equation*}
Hence, once we introduce blow-ups from the perspective of algebraic geometry and line bundles, it will follow that said blow-ups are described by the Néron-Severi group of the manifold.

\subsection{Blow-ups from algebraic geometry}
In order to make a connection with the rest of the work, we will now introduce the blowups of a $K3$ manifold in a perspective more akin to algebraic geometry, and using the notions which we have introduced in the previous chapter. Before going any further, we have to introduce a bit of nomenclature.

The first definition is that of a complex analytic hypersurface
\begin{mydef}
An analytic hypersurface of an analytic variety $X$ is an analytic subvariety $Y\subset X$ of codimension one, namely $\dim(Y)=\dim(X)-1$.
\end{mydef}
Above, an analytic variety is a sort of generalization of complex manifolds, which allows for singularities on the space. More concretely, a complex analytic variety is locally given as the zero locus of a certain (finite) set of holomorphic functions. In other words, it is the union of a certain set of components, which are zero loci of holomorphic functions. In the same fashion, a hypersurface inherits this property from its ambient variety. Hence, every hypersurface $Y$ is the union $\cup_iY_i$ of its so-called \textit{irreducible components}. If the ambient manifold $X$ is compact, this union is finite, but such in general we only have that the union is locally finite\footnote{A locally finite collection of subsets of a topological space (in this case $\{Y_i\}$ as subsets of $X$) is one in which every point $y\in \cup_iY_i$ has a neighbourhood which intersects at most finitely many elements of the collection }. This very property is what allows us to (pointwise) define a formal linear combination of hypersurfaces as follows 
\begin{mydef}
Given $\{Y_i\}_{i\in I}$ a collection of irreducible hypersurfaces of $X$, a divisor $D$ on $X$ is a formal linear combination
\begin{equation}
    D=\sum_{i\in I}a_i[Y_i]
\end{equation}
with $a_i\in \mathbb{Z}$.
\end{mydef}
Note that local finiteness implies that for any $x\in X$ there exists an open neighbourhood of $x$ on which $a_i\neq 0$ and $Y_i\cap U\neq \emptyset$ for only finitely many $i\in I$. In other words, the sum can be defined because it is pointwise finite. 

Given the previous sections, one can see the importance of introducing hypersurfaces and divisors. We would like to understand complex manifolds with tools as simple as possible. One such tool is the idea of line bundles, which as we have seen can be studied by means of their cohomology. Now, if we consider sections of said line bundles, it turns out that hypersurfaces are always given as some global holomorphic zero section of a line bundle. This makes sense, since the fibers of line bundles are by definition one dimensional, and so their complement has codimension one. In fact, we state without proof that there exists a natural group homomorphism between the group of divisors of a complex manifodl $X$, $\text{Div}(X)$, and $\text{Pic}(X)$.

\subsubsection{Blow-ups and line bundles}
Previously, we have seen how to construct the blow-up of $\mathbb{C}^2$ at the origin, by considering tangent lines at the origin. Now, we realize that in order to do so, we had to define the blow-up of our original space in a very similar way to how we later on defined holomorphic line bundles, in the sense that in both cases we had subsets of $\mathbb{CP}^n\times \mathbb{C}^{n+1}$. However, in one case we had a blow-up of $\mathbb{C}^n$ at some point, and in the other a line bundle over $\mathbb{CP}^n$. We will now see how both are related. Again, blow-ups can be performed along complex subvarieties, but for the purposes of this work, we will only consider blow-ups along points.

For a line bundle $\mathcal{O}(-1)\to \mathbb{CP}^n$ we can write down the following diagram of maps
\begin{center}
\begin{tikzcd}
\mathbb{CP}^n \arrow[r, hookrightarrow]\arrow[d]& \mathcal{O}(-1)\arrow[r,hookrightarrow] \arrow[d, "\sigma"] & \mathbb{CP}^n\times\mathbb{C}^{n+1} \arrow[r]\arrow[d] & \mathbb{CP}^n \\
\{0\} \arrow[r, hookrightarrow] & \mathbb{C}^{n+1} \arrow[r, equal] &\mathbb{C}^{n+1} & 
\end{tikzcd}
\end{center}
The line at the top comprises the usual line bundle map, where the fibre over a $l\in \mathbb{CP}^n$ is precisely the line itself. However, instead of working with $\mathbb{CP}^n$ as the base, we can just as well use $\sigma:\mathcal{O}(-1)\to \mathbb{C}^{n+1}$, while keeping the bundle construction from before exactly the same. To make eveerythign match, we define the fibres of $\sigma$ as follows. For $z\neq 0$, $\sigma^{-1}$ is the unique line $l_z\in \mathbb{CP}^n$ that passes through $z\in \mathbb{C}^{n+1}$. For $z=0$, this preimage must be $\sigma^{-1}(0)=\mathbb{CP}^n$, since all lines contain the origin. As a matter of fact, this preimage is the zero section of $\mathcal{O}(1)\to \mathbb{CP}^n$, as it represents $l\mapsto (l,0)$, where we are using the notation from our introduction of line bundles. This allows us to define a blow-up as follows:
\begin{mydef}
The blow-up $\sigma:Bl_0(\mathbb{C}^{n+1})\to \mathbb{C}^{n+1}$ of $\mathbb{C}^{n+1}$ at the origin is the holomorphic line bundle $\mathcal{O}(-1)\to \mathbb{CP}^n$ together with the natural projection $\sigma:\mathcal{O}(-1)\to \mathbb{C}^{n+1}$
\end{mydef}

For the sake of completeness, we will now give the general definition of a blow-up along an arbitrary submanifold.
\begin{mydef}\label{def:blowup}
Let $Y$ be a complex submanifold of $X$. Then there exists a complex manifold $\hat{X}=Bl_Y(X)$, the blow-up of $X$ along $Y$, together with a holomorphic map $\sigma:\hat{X}\to X$ such that $\sigma: \hat{X}\setminus [\sigma^{-1}(Y)]\simeq X\setminus Y$, and $\sigma:\sigma^{-1}(Y)\to Y$ is isomorphic to $\mathbb{P}(\mathcal{N}_{Y\setminus X})\to Y$.
\end{mydef}
Above, $\mathcal{N}_X$ denotes the normal bundle of $X$. The two conditions above may be interpreted as follows. Away from the submanifold along which the blow-up is performed, the procedure does nothing. On the contrary, this submanifold is replaced by (a bundle map from) the projectivization of the normal bundle of that submanifold. If we again consider the particular case of $Y=\{x_0\}$, and use the fact that $\mathbb{P}(\mathbb{C}^n)=\mathbb{CP}^{n-1}$, we recover the definition from above.
\begin{mydef}\label{def:ediv}
The hypersurface $\sigma^{-1}=\mathbb{P}(\mathcal{N}_{Y\setminus X})\subseteq Bl_{Y}(X)$ is called the \textbf{exceptional divisor} of the blow-up $\sigma:Bl_Y(X)\to X$
\end{mydef}
In the particular case that we are dealing with, in which $Y$ is just a point in an $n$ dimensional manifold $x\in X$, then the blow-up replaces $x$ by $\mathbb{CP}^{n-1}$. The exceptional divisor $\sigma^{-1}(x)$ is usually denoted $E$.

\subsection{The ADE classification of $K3$ singularities}
In this short section, we briefly review how singularities are classified on a $K3$ surface. In contrast with other sections, we will not go into much detail, as a lot of these concepts are known from $F$-theory. A standard reference for this is \cite{tasifth}.

In the context of $F$-theory, one considers spacetime compactifications of type IIB string theory of the following form 
\begin{equation}
    \mathcal{M}^{1,9}=\mathbb{R}^{1,9-2n}\times B_n
\end{equation}
with $B_n$ a manifold of (complex) dimension $n$. One further considers 7-branes which fill the $\mathbb{R}^{1,9-2n}$ factor, and that wrap an internal holomorphic cycle $\Sigma_{n-1}\subset B_n$ of codimension one. The 7-branes imply the appeareance of the axion $C_0$ in the theory which, together with the dilaton, combine to give a complex field, the axio-dilaton 
\begin{equation}
    \tau=C_0+ie^{-\varphi}.
\end{equation}
Because of the nontrivial variation of the axio-dilaton around the $D7$-branes, and the fact that $R_{ab}=\nabla_a\nabla_b\varphi\neq 0$, we know that $B_n$ cannot be Calabi-Yau. However, it turns out that this variation of $\tau$ can be used to define an elliptic fibration over $B_n$. An elliptic fibration is one in which each fibre has the structure of a torus with a marked base point (an elliptic curve)
\begin{center}
\begin{tikzcd}
\pi:\mathbb{E}_{\tau} \arrow[r] & Y_{n+1}\arrow[d] \\
 & B_n
\end{tikzcd}
\end{center}
Here, the fact that each fiber $\pi^{-1}(b)$ for $b\in B_{n}$ has a marked point is represented in the requirement that there exists a section $\sigma:B_n\to Y_{n+1}$. The image of the section at each point in the base is the marking of the fiber. It turns out, by calculations regarding the first Chern classes of both $Y_{n+1}$ and $B_n$, that $c_1(Y_{n+1})=0$, and hence $Y_{n+1}$ is Calabi-Yau. In the case of $n=1$, we have that $Y_2=K3$, and $B_{1}=\mathbb{CP}^1$. Note that the $K3$ here is generic. For the remainder of the work, we will stay in this context.

We will now see how singularities may occur on an elliptic fibration, and how they are classified. We will try to be as simple as possible here. Most often, elliptic fibrations are described by a so-called Weierstrass model. The construction is as follows. First, one considers a special type of projective space of dimension 3, namely $\mathbb{P}_{231}$, defined as 
\begin{equation}
    \mathbb{P}_{231}\coloneqq\frac{\mathbb{C}^3\setminus\{0\}}{\sim}, \ \ \ \ \ \ \ \text{where } (x,y,z)\sim (\lambda^2 x, \lambda^3 y, \lambda z) \ \ \text{for }\lambda\in \mathbb{C}^*
\end{equation}
Then, the elliptic curve is defined pointwise as the zero-locus on $\mathbb{P}_{231}$ of 
\begin{equation}\label{eq:Pell}
    P=y^2-(x^3+fxz^4+gz^6)
\end{equation}
Namely, the elliptic curve is defined as the hypersurface $P=0$. Above, $f$ and $g$ are technically understood as sections of a line bundle over the base $B_n$, namely
\begin{align*}
    f&\in \Gamma(B_n, \mathcal{L}^4) \\
    g&\in \Gamma(B_n,\mathcal{L}^6)
\end{align*}
Where powers of a line bundle are understood as coming form tensor products with itself, as in the definition of $\mathcal{O}(k)$. This technicality arises because both $f$ and $g$ are related to the structure constant $\tau$ of the marked torus which defines the elliptic curve. They thus have certain transformation properties under $SL(2,\mathbb{Z})$, the isometry group of the torus, and this is reflected in their definition as sections. In more practical terms, $f$ and $g$ can be thought of as coefficients at each point in the base, which define a particular elliptic curve at that point through \eqref{eq:Pell}. We would also like to stress the fact that this strange looking polynomial is highly non-unique, and depending on certain mathematical details (such as the characteristic of the field over which the curve is defined), it can be brought to simpler forms. However, we would like to stick to the notation used in \cite{tasifth}, as it is a very popular reference.

With the Weierstrass model in mind, we will now see what it means for a singularity to occur, and then we will discuss the sense in which they are classified. We say that a hypersurface $P=0$ is singular if its gradient also vanishes. That is, if we both have that $P=0$ and $dP=0$. Let us see what this means for the Weirstrass model. First, we check that the singularity cannot happen at $z=0$, since then $P=0$ would reduce to $y^2=x^3$, whose gradient is $2y=3x^2$. Therefore, the only possible singularity would be at $x=y=z=0$, which is a point that is not on $\mathbb{P}_{231}$. By the rescaling properties of this projective space, we can thus restrict ourselves to the case where $z=1$, where the Weierstrass equation $P=0$ can be written as 
\begin{equation}
    y^2=x^3+fx+g.
\end{equation}
The left hand side is just a degree 3 polynomial, which in these conditions can be decomposed as 
\begin{equation}
    y^2=K\prod_{i=1}^3(x-x_i),
\end{equation}
where $\{x_i\}$ are the roots. A simple calculation of the gradient as above reveals that a singularity occurs if and only if a root has multiplicity greater than 1. On the other hand, since we are dealing with a degree 3 polynomial, this only happens when its discriminant $\Delta$ vanishes. For this polynomial, the explicit form of the determinant is 
\begin{equation}
    \Delta=4f^3+27g^2.
\end{equation}
There is a bit more to the singularities of the Weierstrass model. In the general case, only $\Delta$ vanishes, and $f$ and $g$ take whichever nonzero values needed for this to happen. However, it can happen that the zero of the discriminant is also a zero of both $f$ and $g$ as functions of the base coordinate. In this case, the singularity is enhanced, and the Weierstrass model becomes itself singular (whereas before, even if $\Delta=0$ at some points, the whole elliptic fibration $Y_{n+1}$ is still smooth).

\subsubsection{The ADE classification}\label{sec:ADEclass}

In principle, an elliptically fibered $K3$ may present many types of singularities, depending on how $\Delta$, $f$, and $g$ vanish. Fortunately, there is a way of classifying these singularities, first introduced in two landmark works by Kodaira \cite{kodaira} and Néron \cite{Neron1964}. We will now briefly discuss this classification, in order to make our latter arguments clearer.

First, we should mention that the classification in question is over smooth, minimal surfaces. The reason for smoothness, is that any singular surface can be made smooth by means of so-called \textit{birational maps}, of which blow-ups (as in definition \ref{def:blowup}) are a particular example. This means that we can "smooth out surfaces in a non-intrusive way". Thus, once smooth surfaces are classified, the non-smooth ones are too. The second condition is minimality. This is somewhat related to the last condition. We say that a smooth surface is minimal whenever it is not birationally equivalent (there is a birational map between them) to another smooth surface. In practical terms, it means that a minimal smooth surface is one which cannot be obtained by blowing up \textbf{smooth points} of some other smooth surface.

Once this technicality is understood, we proceed with the classification of $K3$ singularities. Let $X$ denote some elliptically fibered $K3$ surface, with a certain number of singular points $\{p_1,\ldots, p_n\}$, which we take to be distinct. In order to obtain the resolution of $X$, which we denote by $\overline{X}$, we need to perform $m$ blow-ups, where $n\leqslant m$\footnote{It might be that a certain singular point has a higher order singularity and needs to be blown up more than once}. This, as we have said before, produces a set of $m$ total exceptional divisors $\{E_1,\ldots, E_m\}$. Given the dimensionality, these $E_i$ are nothing but rational curves, isomorphic to $\mathbb{CP}^1$. These are sometimes referred to as \textit{degenerate fibers}, as over the hypersurface defined by $\{\Delta=0\}$, they are degenerate.

Each of the curves that are produced in the blow-up can not only intersect itself, but it can also intersect one or more of the other curves. When we use the term \textit{intersect}, we mean in the usual way in which cycles intersect in homology theory. This is captured in the intersection product, which assigns to any two curves $C_i$ and $C_j$ an integer, $C_i\cdot C_j\in \mathbb{Z}$. We will come back to the intersection structure of a $K3$ in a later section. The punchline, and the basis for the classification of ADE singularities, is the profoundly shocking fact that the intersection structure of these rational curves exactly mimics one of the extended Dynkin diagrams of an ADE Lie algebra. Recall that the Dynkin diagram associated to a Lie algebra is constructed from its simple roots $\{\alpha_i\}$, with $i=1,\ldots, \text{rk}(\mathfrak{g})$, by a pairing $d_{ij}=\langle\alpha_i,\alpha_j\rangle\langle \alpha_j,\alpha_i\rangle$. Because of technical details, which are discussed in \cite{tasifth}, this leads us to identify $E_i\sim (-\alpha_i)$, for $i=1,\ldots, \text{rk}(\mathfrak{g})$.

The above relation can be thought of in a very pictoric and intuitive way. If we use the fact that $\mathbb{CP}^1\simeq \mathbb{S}^2$, the Riemann sphere, a blow-up (or successive blow-ups) replaces singular points with Riemann spheres. The intersection of rational curves from above translates into the intersection of these Riemann spheres at a certain (finite) number of points. Because of the identification of these spheres with the (negative of) the simple roots of some ADE algebra, the spheres arrange themselves literally like in a Dynkin diagram.

There are certainly some subtleties about this classification, and some singularities whose rational curves do not straightforwardly reproduce an extended Dynkin diagram, but can nevertheless still be classified by an ADE group. For our purposes, and the extension of this work, it is enough to have an intuitive idea of how this classification arises, and the mathematics that are relevant to it. We redirect the interested reader to \cite{miranda1989basic, EllipticSchuettShioda, tasifth}, the latter one being the standard refence for $F$-theory in physics. We also note that a full table with the classification of the ADE singulatrities of an elliptically fibered $K3$ can be found on \cite[Table 4.1]{tasifth}.

\section{Recovering the ADE singularities from the blow-ups}
Now that we have introduced the necessary mathematics to properly discuss the topic, and we have a general idea of how singularities occur and are classified on an elliptically fibered $K3$, it is time to go backwards. In the end, our goal is to relate all this to the cobordism conjecture, and try to draw some restrictions on the allowed backgrounds for our theories of quantum gravity. The problem here is that cobordisms are, in a sense, "only sensitive to the topology of a manifold", and not even completely at that. By this we mean that the bordism class of a manifold can be affected by such things as characteristic numbers (the Pontrjagin and Stiefel-Whitney numbers), tangential structures, such as in \cite{andriot2022looking}, etc. However, the main appeal of bordisms is that they may relate two manifolds whose topologies differ completely. One explicit example of this is given in appendix \ref{app:2cp2}, where we check that, while $\mathbb{CP}^2\#\overline{\mathbb{CP}}^2$ and $\mathbb{CP}^2\sqcup\overline{\mathbb{CP}}^2$ are bordant, their homology groups differ.

Throughout the last sections, we have seen that line bundles could be understood through the Picard group of a manifold, and then, through its N{\'e}ron-Severi group. In particular, this leads us to the conclusion that blow-ups, defined in terms of holomorphic line bundles, may be described in the same terms. Since the ADE classification of $K3$ singularities is based on how the exceptional divisors coming from the blow-ups of the singularities intersect, it follows that some aspects of the ADE classification of a $K3$ have to do with the cohomology structure of the manifold. This is really a problem, as these are liable to change under bordisms. In particular, a bordism relation (or the bordism class of a manifold) tells us nothing about its homology structute. It follows that, in order to fully specify the ADE structure, or the exact allowed gauge groups for backgrounds which are allowed by the cobordism conjecture, one needs to add some extra information to what is given by the cobordism conjecture. In particular, from the discussion in the last sections, we would need to specify the way in which the exceptional divisors coming from the blow-ups of the singularities intersect. In this section, we will see what kind of restrictions on the allowed gauge groups are obtainable from this perspective.

\subsection{Setting up the stage}
In previous introductory sections, we have seen how blow-ups can be understood in terms of holomorphic line bundles, and how these can be in turn understood from the N{\'e}ron-Severi group of the manifold $X$ which has been blown up.  Moreover, we have stated that $\text{NS}(X)=H^{1,1}(X,\mathbb{Z})=H^2(X,\mathbb{Z})\cap H^{1,1}(X)$. We can draw two conclusions from here: the first is that $\text{NS}(X)$ will inherit a lattice structure from $H^2(X,\mathbb{Z})$, and the second is that $\rho(X)=\text{rk}(\text{NS}(X))\leqslant 20$. The second conclusion comes from knowing that the dimensionality of the cohomology groups of a $K3$ surface is completely fixed, and encapsulated in the Hodge diamond
\begin{equation}
\begin{matrix}
    &  & 1  &  &    \\
    & 0  &   & 0  &    \\
   1 &  & 20 &  & 1    \\
   & 0 &  & 0 &    \\
   &  & 1  &  &     \\
\end{matrix}
\end{equation}
In particular, we have that $\text{rk}(H^2(X,\mathbb{Z}))=h^{2,0}+h^{1,1}+h^{0,2}=22$, whereas $\text{rk}(H^{1,1}(X))\coloneqq h^{1,1}=20$. The reason why $h^{1,1}$ only provides an upper bound on $\rho(X)$ is that one is allowed to consider complex structure deformations on a $K3$. These change the complex structure of the manifold, which in turn modifies the decomposition \eqref{eq:hodgedec}. Thus, elements of $\text{NS}(X)$ could in principle pick up $(2,0)$ or $(0,2)$ components, and hence not belong to $\text{NS}(X)$ anymore. Thus, the best that we can say about the rank of the Néron-Severi group is that $\rho(X)\leqslant 20$.

Now we come to discussing lattice structures, and how $\text{NS}(X)$ inherits one. This is a consequence of the fact that it is a subgroup of $H^2(X,\mathbb{Z})$, which can itself be given the structure of an \textit{even, self-dual} lattice. In general, we recall that a lattice is essentially a module\footnote{For present purposes, one could substitute "module" by "vector space"} over $\mathbb{Z}$ equipped with a non-degenerate symmetric bilinear form
\begin{equation}
    \langle \cdot, \cdot \rangle: L\times L\to \mathbb{Z}.
\end{equation}
We now proceed to explain how this comes about. Because of Poincaré duality, and the dimension of the manifolds we are working with, this can be seen either form the point of view of homology or cohomology. From the above discussion, we know that $H^2(X,\mathbb{Z})\simeq \mathbb{Z}^{22}$ as a group, so as a module it will have finite rank. Because we are working with the so-called \textit{middle cohomology}\footnote{Namely, we have a manifold of dimension $d=4k$, and the (co)homology groups taken into account are $H^{2k}$.}, the cup product of forms (which in the case of de Rham cohomology is the wedge product) endows $H^2(X, \mathbb{Z})$ with a symmetric, bilinear form 
\begin{equation}\label{eq:intersecc}
\begin{split}
    \langle\cdot,\cdot\rangle:H^2(X,\mathbb{Z})\times H^2(X,\mathbb{Z})&\to H^4(X,\mathbb{Z})\simeq \mathbb{Z} \\
    (\alpha,\beta)&\mapsto \langle\alpha,\beta\rangle =\int_X\alpha\wedge \beta
\end{split}
\end{equation}
where $H^4(X,\mathbb{Z})\simeq \mathbb{Z}$ follows form Poincaré duality. The above defined pairing is usually referred to as \textit{the intersection pairing}. In the dual homology picture, this pairing is precisely the one given by the intersection number, where the pairing between two cycles $a,b\in H_2(X,\mathbb{Z})$ is given by
\begin{equation*}
   a\cdot b=\#(a\cap b),
\end{equation*}
where the symbol $\#$ denotes the cardinality of the set. With any of these two dual pairings, we can endow $H^2(X\mathbb{Z})$ with the structure of a lattice, which we shall denote by $\Lambda_{K3}$. We have now come full circle. The starting point of this work is precisely that the Hirzebruch signature theorem states that the signature of this lattice is $\sigma(K3)=-16$. Since its rank is 22, the signature of the lattice must decompose as $(3,19)$.

The remaining properties of $\Lambda_{K3}$ can be seen as follows. From Poincaré duality, we know that for any basis $\{e_i\}$ of 2-forms, we can find a dual basis $\{e^*_j\}$ such that 
\begin{equation*}
    e_i\cdot e^*_j=\delta_{ij}
\end{equation*}
Again, by duality, we were able to define a basis for $H^{2}(X,\mathbb{Z})$ from $\{e_i\}$. Thus, $\Lambda_{K3}$ is self-dual. Lastly, $\Lambda_{K3}$ is even, meaning that for any $\alpha\in H^2(X,\mathbb{Z})$, we have that  
\begin{equation}
    \langle\alpha,\alpha\rangle\in 2\mathbb{Z}.
\end{equation}
In the dual picture, for $a\in H_2(X,\mathbb{Z})$, we have that 
\begin{equation}\label{eq:evenlat}
    a^2\coloneqq a\cdot a\in 2\mathbb{Z}.
\end{equation}
Rather surprisingly, this comes from the fact that $K3$ is spin, and thus $c_1(K3)^2=c_2(K3)=0$ (mod 2). This is just a restatement of the fact that the first and second Stiefel-Whitney classes $w_1$ and $W_2$ vanish. The proof of this is complicated, and relies on Wu's formula.

The requirement that $\Lambda_{K3}$ has to be even and self-dual severely restrict the possibilities for the lattice. Combining this with the signature and rank of the lattice, we know from intersection theory that the only (unique up to isometries with respect to the inner product defined above) choice is 
\begin{equation}\label{eq:k3lattice}
    \Lambda_{K3}=(-E_8)^{\oplus 2}\oplus H^{\oplus 3}
\end{equation}
Above, $E_8$ denotes the root lattice associated to the Lie algebra of the same name, and $H$ represents the unique rank 2 hyperbolic lattice. With this explicit form of the $K3$ lattice, its intersection form, which governs the inner product defined above, is box-diagonal. Two of these boxes are $8\times 8$ matrices given by minus the Cartan matrix associated to $E_8$ (hence the minus sign in \eqref{eq:k3lattice}), and the three remaining boxes each contain the intersection form associated to $H$, which is the following $2\times 2$ matrix\footnote{Note that we are denoting both the lattice and the intersection form by the same name}:
\begin{equation}\label{eq:HypLattice}
    H=\begin{pmatrix}
     0 & 1 \\
     1 & 0
    \end{pmatrix}
\end{equation}

After having briefly described $\Lambda_{K3}$, the lattice defined by $H^{2}(X,\mathbb{Z})$, we move on to seeing how $\text{NS}(X)\subset H^{2}(X,\mathbb{Z})$ inherits a lattice structure. As we have stated before, the possibility of deforming the lattice structure of the surface means that we can only give an upper bound on the rank of this lattice, namely $\rho(X)\leqslant 20$. Because of the natural inclusion, the symmetric bilinear pairing on $\text{NS}(X)$ is exactly the same as the one defined on $H^{2}(X,\mathbb{Z})$ by the intersection product, but just restricted to $\text{NS}(X)$. The only thing left that we can determine about the Néron-Severi lattice is the signature of its intersection form. The answer to this is given by the Hodge index theorem \cite[Corollary 3.3.16]{huybrechts}, which we repeat here
\begin{prop}
Let $X$ be a compact K\"ahler surface. Then the intersection pairing \eqref{eq:intersecc} has index $(2h^{2,0}+1, h^{1,1}-1)$. Restricted to $H^{1,1}$ it is of index $(1,h^{1,1}-1)$
\end{prop}
We will not provide a proof, but just note in passing that the offset of $1$ in both indices has its roots in the appeareance of the K\"ahler class $[\omega]\in H^{1,1}(X)$. As a consequence of this, given that $\text{NS}(X)=H^{2}(X,\mathbb{Z})\cap H^{1,1}(X)$, we have that the signature of the Néron-Severi lattice is $(1,\rho-1)$

\subsection{Introducing the blow-ups}

The only information that the cobordism conjecture gives us, working in $\Omega_4^{SO}$, is that for a generic $K3$ to be a valid background, it must be accompanied by a certain number of copies of $\mathbb{CP}^2$, which we interpret (through passing from disjoint unions to connected sums under the bordism relation) as coming from a number of blow-ups performed on ADE singularities on the original $K3$. We note in passing that a connected sum  $\mathbb{CP}^2\#\stackrel{n)}{\ldots}\#\mathbb{CP}^2$ is interpreted as coming from resolution of $m$ points $p_1,\ldots,p_m$ in the manifold, for $m\leqslant n$. While we know exactly how many blow-ups are needed, we will keep the number arbitrary for the time being. Now, we need to see how these blow-ups might relate to the ADE structure of the original $K3$. In the following, we will name the original $K3$ (with singularities) $X$, and the resolved $K3$ will be called $\hat{X}$.

As we have seen before, blow-ups are understood in terms of line bundles. For each blow-up, we produce an exceptional divisor $E_i$ on the blown-up manifold $\hat{X}$, as in definition \ref{def:ediv}. Each of these exceptional divisors is given by a $\mathbb{CP}^1$, and so they are rational curves. Thus, if we blow up $n$ points, we will end up with $n$ exceptional divisors $E_1,\ldots,E_n$, which should fit in the Néron-Severi lattice of the blow-up surface, $NS(\hat{X})$. This, in particular, forces them to satisfy \eqref{eq:evenlat}. The fact that they do is guaranteed by the so-called \textit{adjunction formula}. This formula deals with the canonical bundle of a manifold and that of a hypersurface which is embedded in said manifold. In particular, for a curve $C$ of genus $g$ embedded on a surface $Y$ with canonical divisor\footnote{The canonical divisor is the divisor associated to the canonical (line) bundle of a manifold.} $K$, it states that 
\begin{equation}
    C\cdot\left(C+K\right)=2g-2=-\chi(C)
\end{equation}
However, since we are dealing with a Calabi-Yau, the canonical bundle is by definition trivial, so its canonical divisor is $K=0$, and thus we have that the self-intersection of a curve is given by its genus $g$. Since $\mathbb{CP}^1\simeq \mathbb{S}^2$, $g=0$, and hence
\begin{equation}
    E_i\cdot E_i=-2
\end{equation}
so all of these exceptional divisors fit nicely into $\text{NS}(\hat{X})\subset\Lambda_{K3}$.

In summary, we have seen that all of the exceptional divisors are algebraic curves of self-intersection $-2$, which fit into the Néron-Severi latice of $\hat{X}$. Therefore, they themselves span an even sublattice $M\coloneqq \text{Span}_{\mathbb{Z}}\{E_1,\ldots,E_n\}\subset \text{NS}(X)$. The ADE classification that we introduced in section \ref{sec:ADEclass}, which had to do with how these exceptional divisors intersect, is now reflected in the fact that $M$ must be isometric to a direct sum of lattices of the same ADE type as the original singularities (i.e. given by the Cartan matrices of said algebras). Since isometry type of a lattice is determined by its intersection form, this is the same as requiring the singular fibers to intersect in a way that reproduces the extended Dynkin diagram of the ADE algebras. All in all, the lattice $M$ spanned by the exceptional divisors induces a decomposition of the Néron-Severi lattice as $\text{NS}(X)=M\oplus M^{\perp}$, whose rank is $n\leqslant \text{rk}(M^{\perp})\leqslant 20$.

There is however a problem with this line of thought, and it is the fact that \textit{we have no information about the ADE type of the singularities}. Said in another way, we do not know how the exceptional divisors intersect, or the isometry type of $M$. The cobordism conjecture says nothing about this, but only about the number of exceptional divisors that are present. Thus, it follows that we would need to add information by hand in order to distinguish the type of ADE algebras to which the singularities correspond. The way of doing this might be through \textit{lattice polarizations}\footnote{D.M. would like to thank Michael Schultz for discussion on these notions}.

\subsection{Lattice polarizations}
The answer to endowing the sub-lattice $M$ with the information necessary to recover the ADE singularities might come from \textit{lattice polarizations}. The idea underlying this is to specify the way in which the Néron-Severi lattice embeds into $\Lambda_{K3}$. In turn, this specifies the intersection structure of the exceptional divisors, and hence the ADE singularities. 
\begin{mydef}
Let $L$ be an even, non degenerate lattice of signature $(1,t)$. An $L$-polarized $K3$ surface is a $K3$ surface $Y$ together with a primitive lattice embedding $i:L\hookrightarrow \text{NS}(X)$ such that $i(L)$ contains a pseudo-ample class.
\end{mydef}
\begin{mydef}
A sublattice $L\subseteq M$ is said to be \textit{primitive} if the quotient $M/L$ is a torsion-free Abelian group. A lattice embedding $i:L\hookrightarrow M$ is said to be \textit{primitive} if $i(L)$ is a primitive sublattice of $M$. Presently, this implies that the short exact sequence of Abelian groups $0\to L \to M \to M/L\to 0$ always splits, and hence there is an isomorphism of Abelian groups 
\begin{equation}
    M\simeq L\oplus (M/L)
\end{equation}
\end{mydef}
As  we have seen in the previous section, the exceptional divisors span an even sublattice $M\subset \text{NS}(X)$. This then leads to a decomposition of the Néron-Severi lattice by considering the complement of $M$, namely $\text{NS}(X)=M\oplus M^{\perp}$. Thus, we could use the tool of lattice polarizations, and take $L$ to be (isometric to) a direct sum of lattices of ADE type. In this way, $L$ would play the role of $M$ in the decomposition of $\text{NS}(X)$. Hence, the lattice polarization would specify the intersection structure of the exceptional divisors, and in turn would characterize the ADE type of the singularities on the $K3$ surface. This whole procedure allows us to complement the cobordism conjecture with the information that it does not provide. We will come back to this particular case in the next section.

Now, we need a way to find out which ADE lattices can polarize a $K3$ surface, and if, further, they can be realized as coming from the singularities of an elliptically fibered $K3$. This in fact not an easy question, since we have the requirement that $i(L)$ must contain a pseudo-ample class (which we have in fact not defined). One of the first conditions that we can anticipate for our polarizing lattice is that, due to the requirement that its signature must be that of a so-called \textit{hyperbolic} lattice, namely $(1,t)$. Hence, purely positive-definite or negative-definite lattices are off the table, as those would have signatures $(n,0)$ and $(0,n)$, respectively. One way to include them would be to consider their direct sum with $H$, the unique hyperbolic rank 2 lattice which we introduced a few paragraphs above. Namely, we would look for lattices $L$ which decomposed as $L=H\oplus M$, with $M$ some positive or negative definite lattice. In fact, it turns out that a lattice polarization is equivalent to an elliptic fibration whenever the polarizing lattice $L$ admits a decomposition $L=H\oplus M$ (or is isometric to one such decomposition).

One publication that might shed some light with this issue is  \cite[Table 1, pg. 1434]{NikulinPolarization}. In particular, the table that we have cited contains all of the hyperbolic, even, 2-elementary lattices which admit a primitive embedding into $\Lambda_{K3}$. A 2-elementary lattice $L$ is one in which its discriminant group, $L^*/L$ satisfies
\begin{equation}
    L^*/L\simeq (\mathbb{Z}/2\mathbb{Z})^a
\end{equation}
Here, $L^*$ denotes the dual lattice, and we have used the fact that due to the symmetric bilinear product there is an inclusion $L\subseteq L^*$. In the publication we have cited, these kind of lattices are characterized by three integers, namely $(\text{rk}(L), a(L), \delta_L)$. The first is the rank of the lattice, the second one is related to the discriminant group as above\footnote{Note that in the publication this parameter is denoted by $l(A_S)$, with $S$ being the lattice in question.}, and $\delta_L$, defined as follows. For an even, 2-elementary lattice $L$, consider its dual, $L^*$, and consider the pairing on $L^*$ defined by the extension of the pairing on $L$ by means of the inclusion $L\subseteq L^*$. Then 
\begin{equation*}
    \delta_L=\begin{cases}
    0 \ \text{if } \langle \alpha, \alpha\rangle\in\mathbb{Z} \ \forall \ \alpha\in L^* \\
    1 \ \text{otherwise}
    \end{cases}
\end{equation*}
With the above defined data, one is able to determine which lattices can and cannot be embedded into $\Lambda_{K3}$. This is done for example in the original reference \cite{NikulinPolarization}, but also in the more recent \cite{JacobianK3}. We will work off of the latter reference from now on, as it contains a broader set of data.

\subsection{Restrictions on the allowed ADE singularities from lattice polarizations and the cobordism conjecture.}
We will now try to put everything together. Firstly, as we have said, the cobordism conjecture states that in order for a $K3$ surface to be a valid background for a theory of quantum gravity, it must come accompanied by a connected sum of 16 copies of $\mathbb{CP}^2$. Namely, the $K3$ surface that we must consider is $\hat{X}=X\# \mathbb{CP}^2\#\stackrel{16)}{\ldots}\#\mathbb{CP}^2$, with $X$ being our original $K3$ surface. We used the fact that a connected sum of some manifold $X$ with copies of $\mathbb{CP}^2$ is diffeomorphic to $X$ blown up a certain number of times to interpret this as the condition that the original $K3$ surface must be blown up exactly 16 times. From our discussion of line bundles and blow-ups, we know that this produces 16 exceptional divisors, which span a lattice $M=\text{Span}_{\mathbb{Z}}\{E_1,\ldots, E_{16}\}$. The intersection structure of these divisors then determines the ADE type of the singularities that were blown-up. However, the cobordism conjecture contains no such information about the intersection structure of the divisors. Hence, we resorted to lattice polarizations in order to try to extract some conditions on the allowed ADE type of the singularities.

As we have said before, from the perspective of lattice polarizations, the presence of a factor of $H$, the unique hyperbolic rank 2 lattice, whose intersection form is given by \eqref{eq:HypLattice} is equivalent to the fact that the $K3$ surface is elliptically fibered. Moreover, as we have preiously stated, the fact that the signature of this lattice is $(1,1)$ allows for polarizing lattices which are not hyperbolic, but rather positive or negative definite, by considering their direct sum with $H$. Thus, we may break the polarizing lattice down as $L=H\oplus M$, where $M$ is the lattice which determines the ADE structure of the singularities. Then, the cobordism conjecture states that $\text{rk}(M)=16$, and therefore the polarizing lattice must have $\text{rk}(L)=18$.

\textit{A priori}, the fact that the rank of the ADE part of the polarizing lattice should have rank 16 should start ringing some bells. Indeed, there are two such lattices which are of great importance for string theory, as they are the only two lattices which are allowed as gauge groups in heterotic string theory. If we further require that the automorphism group of the resulting (resolved) $K3$ surface be finite, we can check in \cite[Table 3]{JacobianK3} that the two possibilities for the polarizing lattice $L=H\oplus K$ are indeed
\begin{equation}
    K=\begin{cases}
    (-E_8)\oplus (-E_8), \\
    -D_{16}.
    \end{cases}
\end{equation}
Note that the minus signs come form the fact that the curves over which the polarized lattice acts must be of self-intersention $-2$. Given the structure of $\Lambda_{K3}$, it was expected that, if the ranks matched, we could embed two copies of $E_8$ into it.
Lastly, we comment on one piece of data present on the table we have referenced. In this publication, an explicit difference is made between $K$, the polarizing lattice itself, and $K^{\text{root}}$, the sub-lattice spanned by the roots of $K$. Namely, $K^{\text{root}}\subseteq K$ is the lattice spanned by the elements of self intersection $-2$. The quotient between the two is denoted by\footnote{We note in passing that this quotient is related to the Mordell-Weil group of the elliptic fibration.} $\mathcal{W}\coloneqq K/K^{\text{root}}$. For the two cases above, we have that for $(-E_8)\oplus (-E_8)$, $\mathcal{W}=1$, and hence both lattices are one and the same. However, for $D_{16}$, $\mathcal{W}=\mathbb{Z}_2$. Since $D_{16}$ is the lattice defined by the Lie algebra corresponding to $SO(32)$, the root lattice is in fact given by $SO(32)/\mathbb{Z}_2$. Thus, the two gauge groups associated to the singularities for a $K3$ surface under the assumptions of the cobordism conjecture (assuming that the automorphism group is finite) are precisely the only ones allowed for the heterotic string.

To close this chapter off, we note that the rank of the polarizing lattice ($\text{rk}(L)=18$) does not exhaust the upper bound on the rank of $\text{NS}(X)$, which is $\rho(X)\leqslant 20$. Hence, there are still some possibilities to "fill out" the remainder of the polarizing lattice. 

\section{Other scenarios of interest}
Staying in the $SO$-structure, we could also consider $\Omega_8^{SO}$, which is generated by $C\mathbb{P}^4$ and $C\mathbb{P}^2\times C\mathbb{P}^2$. This could give us more freedom to play with the two generators to kill the cobordism class.

If we decide to leave the $SO$-structure, we could apply the same logic to more complicated structures, such as $\Omega_d^{\text{Spin}}$, $\Omega_d^{\text{String}}$, etc. As shown in \cite{structcob}, these particular ones mentioned can be obtained by further constraining the vanishing of certain charactersitic numbers. For example, $\Omega_d^{\text{Spin}}$ is obtained from $\Omega^{SO}_d$ by requiring that $w_2=0$, as expected.

\newpage
\section{Conclusions and Outlook}
In our work, we have applied the techniques of the Ricci flow equations with surgery, as introduced by Hamilton and later brought to full fruition by Perelman, and tried to combine them with the techniques of algebraic geometry in order to study possible implications of the cobordism conjecture. After having given a reasonable introduction to the relevant mathematical tools, it was argued -and explicitly shown for the simple cases of Einstein manifolds and Ricci solitons- that Ricci flow preserves Ponryagin and Stiefel-Whitney numbers. This holds even in the presence of a "neckpinch", as the proper treatment of it under the Ricci flow is to perform a particular type of surgery: substituting the pinched neck by two smooth caps. Since these characteristic numbers unambiguously fix an element of $\Omega_4^{SO}$, the above statement implies that Ricci flow preserves the corresponding cobordism class. It is nevertheless interesting and reasonable, from a more physical perspective, to investigate the effects of completely removing one of the two components resulting from the surgery. In general, this can imply the necessity of balancing such a modification out by the introduction of defects, that mimic the contribute of the removed component to the overall cobordism class. Therefore, focusing on the specific string-theory inspired example of a $K3$, it was argued that its cobordism class can be trivialised by adding to it the connected sum of 16 copies of $\mathbb{CP}^2$.

In the later sections of the present work, the trivialization of the $K3$ class in $\Omega^{SO}_4$ is studied from two complementary perspectives. On the one hand, we tried to apply the Ricci flow equations to the resulting manifold, in order to see if we could recover the topological defects (as the shrinking of parts of the manifold to points) by means of the flow. That is, to see if the Ricci flow is the bridge connecting the perspectives of adding connected sums of manifolds and adding topological defects. This was done by means of a metric constructed on $\mathbb{CP}^2\#\stackrel{n)}{\ldots}\#\mathbb{CP}^2$, obtained by Claude Lebrun as a generalization of the so-called Burns metric. We give arguments as to why this manifold shrinks under the Ricci flow, thus producing singularities and infinite distance limits in the moduli space of space-time metrics, of the kind which are discussed in the context of the Swampland distance conjecture. 

On the other hand, after the realization that taking connected sums of a manifold with $\mathbb{CP}^2$ is equivalent to considering a blow-up of said manifold, and that the singularities of $K3$ surfaces are classified by the ADE Lie Algebras, we try to complement the information given by the cobordism conjecture. Namely, we realize that we can recover the singularities by flowing the $\mathbb{CP}^2$s that they produce, but we realize that the cobordism conjecture contains no information whatsoever regarding the ADE nature of these singularities. Hence, this is a piece of data that we must introduce by hand, and might tell us about possible restrictions on the ADE nature of the singularities. After introducing the necessary tools from algebraic geometry, and reviewing how the ADE classification comes about, we propose that this extra information might be added by means of lattice polarizations of $K3$ surfaces. Upon doing so, and given the conditions imposed by the cobordism conjecture, we recover the only two groups which are allowed as gauge groups for the heterotic string. 

While the technique of lattice polarizations is best understood for $K3$ surfaces, this work could be expanded by considering more complicated structures, such as those on $\Omega_d^{\text{spin}}$ or $\Omega_d^{string}$. In this context, one could find which manifolds are needed to trivialize the bordism class, and try to apply the Ricci flow equations to them. This would be easiest when staying in the $SO$-structure, since we know that $\Omega_8^{SO}$ is generated by $\mathbb{CP}^4$ and $\mathbb{CP}^2\times \mathbb{CP}^2$. Not only are these manifolds well-known, but also the fact that we have more than one generator might allow us more freedom to encounter different scenarios for the trivialization of a given bordism class.
\section*{Acknowledgments}
The work of D.L. is supported by the Origins Excellence Cluster and by the German-Israel-Project (DIP) on Holography and the Swampland.

\newpage
\section{Appendix: The (co)homology of $\mathbb{CP}^{2}\#\mathbb{CP}^2$}\label{app:2cp2}
In this somewhat more technical section, we will work out the homology and cohomology groups of $\mathbb{CP}^2\#\mathbb{CP}^2$. We will use a collection of results which are fundamental in the context of algebraic topology, but nonetheless require some degree of familiarity with the topic. Due to this technical nature, we will not introduce these concepts here, but rather point the reader to standard references such as \cite{HatcherAT}.

To begin the discussion, we note that the orientation of an $n$-dimensional manifold does not affect its homological structure, but rather only the choice of generator of the top homology group $H_n$. We refer to this generator as the fundamental class of the manifold. Hence, in particular, we have that $H_p(\mathbb{CP}^2)=H_p(\overline{\mathbb{CP}}^2)$, for $p=1,\ldots, 4$, all other homology groups being automaticall trivial for dimensional reasons. Furthermore, as a consequence of the CW structure of $\mathbb{CP}^2$, we have that 
\begin{equation}
    H_p(\mathbb{CP}^2)=\begin{cases}
    \mathbb{Z} \hspace{1cm} &\text{if } p=0,2,4 \\
    0 & \text{else}.
    \end{cases}
\end{equation}
However, in order to simplify calculations, we will not use the above (singular) homology, but the so-called reduced homology, which we denote by $\Tilde{H}_p$. In this context of path connected spaces, this simply means that we can set $\Tilde{H}_p(\mathbb{CP}^2)=0$, and hence 
\begin{equation}
    \tilde{H}_p(\mathbb{CP}^2)=\begin{cases}
    \mathbb{Z} \hspace{1cm} &\text{if } p=2,4 \\
    0 & \text{else}.
    \end{cases}
\end{equation}
In the following, we will drop the tilde, as we will only be considering reduced homology groups.

The actual computation of the homology groups of $\mathbb{CP}^2\#\overline{\mathbb{CP}}^2$ will be a two-step process fo sorts. First, we will consider the long exact sequence of the pair $(\mathbb{CP}^2\#\overline{\mathbb{CP}}^2, \mathbb{S}^3)$. In said pair, all of the points on the $\mathbb{S}^3$ over which we glue to define the connected sum are identified as one, and thus $(\mathbb{CP}^2\#\overline{\mathbb{CP}}^2, \mathbb{S}^3)\simeq \mathbb{CP}^2\vee\overline{\mathbb{CP}}^2$, where $\vee$ denotes the one-point union of two spaces. Once this is done, we will use a simple Mayer-Vietoris argument that will allow us to directly compute the homology groups of our original space in terms of those of the original $\mathbb{CP}^2$ factors.

The long exact sequence of the pair $(\mathbb{CP}^2\#\overline{\mathbb{CP}}^2, \mathbb{S}^3)$ is given by 
\begin{equation}\label{eq:lespair}
    \ldots\longrightarrow H_p(\mathbb{CP}^2\#\overline{\mathbb{CP}}^2, \mathbb{S}^3)\stackrel{\partial_*}{\longrightarrow}H_p(\mathbb{S}^3)\longrightarrow H_p(\mathbb{CP}^2\#\overline{\mathbb{CP}}^2)\longrightarrow H_p(\mathbb{CP}^2\#\overline{\mathbb{CP}}^2, \mathbb{S}^3)\longrightarrow\ldots
\end{equation}
Because the homology groups of spheres are
\begin{equation}
    H_p(\mathbb{S}^n)=\begin{cases}
    \mathbb{Z} \hspace{1cm} &\text{if } p=n \\
    0 & \text{else}.
    \end{cases}
\end{equation}
Then, for $p\neq 3,4$, we have that 
\begin{equation}
    \ldots \longrightarrow 0\longrightarrow H_p(\mathbb{CP}^2\#\overline{\mathbb{CP}}^2)\stackrel{\simeq}{\longrightarrow}H_p(\mathbb{CP}^2\#\overline{\mathbb{CP}}^2,\mathbb{S}^3)\longrightarrow 0 \longrightarrow\ldots
\end{equation}
Therefore, by exactness of the sequence, the map in the middle is an isomorphism, so $H_p(\mathbb{CP}^2\#\overline{\mathbb{CP}}^2)\simeq H_p(\mathbb{CP}^2\#\overline{\mathbb{CP}}^2, \mathbb{S}^3)\simeq H_p(\mathbb{CP}^2\vee\overline{\mathbb{CP}}^2)$ for $p\neq 3,4$. For the special case $p=3,4$, the interesting piece of the long exact sequence looks like 
\begin{equation}\label{eq:lespairspicy}
\begin{split}
    \ldots\longrightarrow 0 &\longrightarrow H_4(\mathbb{CP}^2\#\overline{\mathbb{CP}}^2)\longrightarrow H_4(\mathbb{CP}^2\#\overline{\mathbb{CP}}^2,\mathbb{S}^3)\longrightarrow\mathbb{Z}\longrightarrow \\
    &\longrightarrow H_3(\mathbb{CP}^2\#\overline{\mathbb{CP}}^2)\longrightarrow H_3(\mathbb{CP}^2\#\overline{\mathbb{CP}}^2,\mathbb{S}^3)\longrightarrow 0\longrightarrow \ldots
    \end{split}
\end{equation}
For the moment, we leave this as is.

Now that we have related the homology groups of $\mathbb{CP}^2\#\overline{\mathbb{CP}}^2$ to those of $\mathbb{CP}^2\bigvee\overline{\mathbb{CP}}^2$, we can calculate the latter by a simple Mayer-Vietoris argument. Let $U,V\subseteq \mathbb{CP}^2\bigvee\overline{\mathbb{CP}}^2$ be two subsets such that 
\begin{equation}
    \mathbb{CP}^2\bigvee\overline{\mathbb{CP}}^2=\text{Int}(U)\cup\text{Int}(V),
\end{equation}
 and that $U\cap V\neq \emptyset$. To ensure that both of these conditions are satisfied, we pick $U$ and $V$ to be neighbourhoods of $\mathbb{CP}^2$ and of $\overline{\mathbb{CP}}^2$ in $\mathbb{CP}^2\vee\overline{\mathbb{CP}}^2$ respectively. Furthermore, we take the neighbourhood small enough so that $U\cap V$ is contractible to a point (said point being the one that defines the one-point union). This will imply that the homology groups of $U\cap V$ will all be trivial.
 
 In the above conditions we have the so-called Mayer-Vietoris long exact sequence:
 \begin{equation}
     \ldots \longrightarrow H_p(U\cap V)\longrightarrow H_p(U)\oplus H_p(V)\longrightarrow H_p(U\cup V)\longrightarrow H_{p-1}(U\cap V)\longrightarrow \ldots
 \end{equation}
 As already stated before, we have that $U\cap V$ deformation retracts to a point, and hence all of the homology groups are trivial (even $H_0(\text{pt})$, since we are using reduced homology). Furthermore, both $U$ and $V$ deformation retract to $\mathbb{CP}^2$ and $\overline{\mathbb{CP}}^2$ respectively, so in total we have that
 \begin{equation}
     \ldots \longrightarrow H_p(\text{pt})\longrightarrow H_p(\mathbb{CP}^2)\oplus H_p(\overline{\mathbb{CP}}^2)\longrightarrow H_p(\mathbb{CP}^2\vee\overline{\mathbb{CP}}^2)\longrightarrow H_{p-1}(\text{pt})\longrightarrow \ldots
 \end{equation}
 However, since all of the (reduced) homology groups of a point vanish, by exactness of the sequence this implies that 
 \begin{equation}
     H_p(\mathbb{CP}^2)\oplus H_p(\overline{\mathbb{CP}}^2)\simeq H_p(\mathbb{CP}^2\vee\overline{\mathbb{CP}}^2)
 \end{equation}
for all $p\in\mathbb{Z}$.

By the above argument, we can just plug $H_p(\mathbb{CP}^2)\oplus H_p(\overline{\mathbb{CP}}^2)$ into \eqref{eq:lespair}. By previous arguments, we have that for $p=0,1,2$ we have that $H_p(\mathbb{CP}^2\#\overline{\mathbb{CP}}^2)\simeq H_p(\mathbb{CP}^2)\oplus H_p(\overline{\mathbb{CP}}^2)$. However, the really interesting case is that of \eqref{eq:lespairspicy}. To study it, we will use the following two facts \cite[Th. 3.26]{HatcherAT}:
\begin{itemize} 
    \item If $\mathcal{M}$ is an orientable, $n$-dimensional space, then 
    \begin{equation}
        H_n(\mathcal{M})=\mathbb{Z}.
    \end{equation}
    The choice of generator is called the fundamental class of $\mathcal{M}$, and determines its orientation.
    
    \item The connected sum of orientable spaces is again orientable.
\end{itemize}

Combining these two results, we have that \eqref{eq:lespairspicy} can be written as
\begin{equation}
    0\longrightarrow \mathbb{Z}\longrightarrow \mathbb{Z}\oplus\mathbb{Z}\stackrel{\varphi}{\longrightarrow}\mathbb{Z}\stackrel{\psi}{\longrightarrow}H_3(\mathbb{CP}^2\#\overline{\mathbb{CP}}^2)\longrightarrow 0
\end{equation}
The key to the above sequence is precisely the map $\varphi$. If we can show that it is surjective, then it follows by exactness that $\text{Im}(\varphi)=\text{Ker}(\psi)$, and hence $\psi$ is the zero map. Furthermore, and also by exactness, $\psi$ is itself surjective, for the last map is by force the zero map. Thus, from the surjectivity of $\varphi$ it would follow that 
\begin{equation}
    H_3(\mathbb{CP}^2\#\overline{\mathbb{CP}}^2)=0
\end{equation}
This would be the last piece in our calculation of the homology groups of $\mathbb{CP}^2\#\overline{\mathbb{CP}}^2$. The only difference with \eqref{eq:hom2cp2} would be that there we used singular homology, so $H_0(\mathbb{CP}^2\#\overline{\mathbb{CP}}^2)$ is nontrivial.

The fact that $\varphi$ is surjective follows from the particular way in which one constructs the long exact sequence of a pair. Given $A\subseteq X$, $H_p(X,A)$ is defined as the $p$th homology group with respect to the quotient space $X/A$, in which $A$ is identified to a point. Thus, cycles in the pair $(X,A)$ either have no boundary, or their boundary lies in $A$. To construct the sequence, the map $\partial_*:H_p(X,A)\to H_{p-1}(A)$ is then defined by precisely taking the boundary in $A$.

Now consider the fundamental class of the pair $(\mathbb{CP}^2\#\overline{\mathbb{CP}}^2, \mathbb{S}^3)$. It is that of $\mathbb{CP}^2\#\overline{\mathbb{CP}}^2$, but with a $\mathbb{S}^3$ identified. Therefore, in the total space, this fundamental class has a boundary, which is precisely the fundamental class of $\mathbb{S}^3$. In other words, one of the generators of the $\mathbb{Z}\oplus \mathbb{Z}$ (depending on which $\mathbb{CP}^2$ we take as a reference) gets mapped to the generator of $H_3(\mathbb{S}^3)$, thus making $\varphi$ surjective.

\newpage
\printbibliography[
heading=bibintoc,
title={References}
]

@book{milnor,
    title = {Characteristic classes},
    author = {John W. Milnor, James D. Stasheff},
    isbn = {0-691-08122-0},
    series = {Annals of mathematics studies, no. 76},
    year = {1974},
    publisher = {Princeton University Press}
}

@article{GilMedrano1991THERM,
  title={THE RIEMANNIAN MANIFOLD OF ALL RIEMANNIAN METRICS},
  author={Olga Gil-Medrano and Peter W. Michor},
  journal={Quarterly Journal of Mathematics},
  year={1991},
  volume={42},
  pages={183-202}
}

@article{lust2019ads,
  title={AdS and the Swampland},
  author={L{\"u}st, Dieter and Palti, Eran and Vafa, Cumrun},
  journal={Physics Letters B},
  volume={797},
  pages={134867},
  year={2019},
  publisher={Elsevier}
}

@book{lee,
    title = {Introduction to smooth manifolds},
    author = {John M. Lee},
    isbn = {978-1-4419-9981-8},
    series = {Graduate Texts in Mathematics},
    year = {2013},
    publisher = {Springer}
}

@article{Hamilton4Man,
  title={Four-manifolds with positive isotropic curvature},
  author={Richard S. Hamilton},
  journal={Communications in Analysis and Geometry},
  year={1997},
  volume={5},
  pages={1-92}
}

@book{HatcherAT,
      author        = "Hatcher, Allen",
      title         = "{Algebraic topology}",
      publisher     = "Cambridge Univ. Press",
      address       = "Cambridge",
      year          = "2000",
}

@article{mcnamara2019cobordism,
  title={Cobordism classes and the swampland},
  author={McNamara, Jacob and Vafa, Cumrun},
  journal={arXiv preprint arXiv:1909.10355},
  year={2019}
}

@ARTICLE{nCP2s,
    author = {Claude Lebrun},
    title = {Explicit self-dual metrics on $CP^2\# \ldots \#CP^2$},
    journal = {J. Differential Geom},
    year = {1991},
    pages = {223--253}
}

@book{huybrechts,
  title={Complex Geometry: An Introduction},
  author={Huybrechts, D. and Springer-Verlag (Berlin)},
  isbn={9783540212904},
  lccn={2004108312},
  series={Universitext (Berlin. Print)},
  url={https://books.google.de/books?id=sWbd0rE3mhIC},
  year={2005},
  publisher={Springer}
}

@article{vafa2005string,
  title={The string landscape and the swampland},
  author={Vafa, Cumrun},
  journal={arXiv preprint hep-th/0509212},
  year={2005}
}

@article{Palti:2019pca,
    author = "Palti, Eran",
    title = "{The Swampland: Introduction and Review}",
    eprint = "1903.06239",
    archivePrefix = "arXiv",
    primaryClass = "hep-th",
    reportNumber = "MPP-2019-53",
    doi = "10.1002/prop.201900037",
    journal = "Fortsch. Phys.",
    volume = "67",
    number = "6",
    pages = "1900037",
    year = "2019"
}

@article{vanBeest:2021lhn,
    author = "van Beest, Marieke and Calder\'on-Infante, Jos\'e and Mirfendereski, Delaram and Valenzuela, Irene",
    title = "{Lectures on the Swampland Program in String Compactifications}",
    eprint = "2102.01111",
    archivePrefix = "arXiv",
    primaryClass = "hep-th",
    month = "2",
    year = "2021"
}

@article{Brennan:2017rbf,
    author = "Brennan, T. Daniel and Carta, Federico and Vafa, Cumrun",
    title = "{The String Landscape, the Swampland, and the Missing Corner}",
    eprint = "1711.00864",
    archivePrefix = "arXiv",
    primaryClass = "hep-th",
    reportNumber = "IFT-UAM-CSIC-17-105",
    doi = "10.22323/1.305.0015",
    journal = "PoS",
    volume = "TASI2017",
    pages = "015",
    year = "2017"
}

@article{Ooguri:2006in,
    author = "Ooguri, Hirosi and Vafa, Cumrun",
    title = "{On the Geometry of the String Landscape and the Swampland}",
    eprint = "hep-th/0605264",
    archivePrefix = "arXiv",
    reportNumber = "CALT-68-2600, HUTP-06-A017",
    doi = "10.1016/j.nuclphysb.2006.10.033",
    journal = "Nucl. Phys. B",
    volume = "766",
    pages = "21--33",
    year = "2007"
}

@article{Kehagias:2019akr,
    author = {Kehagias, Alex and L{\"u}st, Dieter and L{\"u}st, Severin},
    title = "{Swampland, Gradient Flow and Infinite Distance}",
    eprint = "1910.00453",
    archivePrefix = "arXiv",
    primaryClass = "hep-th",
    reportNumber = "MPP-2019-198, LMU-ASC 32/19, IPhT-T19/133, CPHT-RR055.092019",
    doi = "10.1007/JHEP04(2020)170",
    journal = "JHEP",
    volume = "04",
    pages = "170",
    year = "2020"
}

@article{Bykov:2020llx,
    author = "Bykov, Dmitri and L{\"u}st, Dieter",
    title = "{Deformed $\sigma$-models, Ricci flow and Toda field theories}",
    eprint = "2005.01812",
    archivePrefix = "arXiv",
    primaryClass = "hep-th",
    reportNumber = "MPP-2020-59, LMU-ASC 16/20",
    month = "5",
    year = "2020"
}

@article{DeBiasio:2020xkv,
    author = {De Biasio, Davide and L{\"u}st, Dieter},
    title = "{Geometric Flow Equations for Schwarzschild-AdS Space-Time and Hawking-Page Phase Transition}",
    eprint = "2006.03076",
    archivePrefix = "arXiv",
    primaryClass = "hep-th",
    reportNumber = "LMU-ASC 19/20, MPP-2020-80",
    doi = "10.1002/prop.202000053",
    journal = "Fortsch. Phys.",
    volume = "68",
    number = "8",
    pages = "2000053",
    year = "2020"
}

@article{Luben:2020wix,
    author = {L{\"u}ben, Marvin and L{\"u}st, Dieter and Metidieri, Ariadna Ribes},
    title = "{The Black Hole Entropy Distance Conjecture and Black Hole Evaporation}",
    eprint = "2011.12331",
    archivePrefix = "arXiv",
    primaryClass = "hep-th",
    reportNumber = "LMU-ASC 46/20, MPP-2020-210",
    doi = "10.1002/prop.202000130",
    journal = "Fortsch. Phys.",
    volume = "69",
    number = "3",
    pages = "2000130",
    year = "2021"
}

@article{friedan1985nonlinear,
  title={Nonlinear models in 2+ $\varepsilon$ dimensions},
  author={Friedan, Daniel Harry},
  journal={Annals of physics},
  volume={163},
  number={2},
  pages={318--419},
  year={1985},
  publisher={Elsevier}
}

@article{friedan1980nonlinear,
  title={Nonlinear models in 2+ $\varepsilon$ dimensions},
  author={Friedan, Daniel},
  journal={Physical Review Letters},
  volume={45},
  number={13},
  pages={1057},
  year={1980},
  publisher={APS}
}

@article{Polyakov:1975rr,
    author = "Polyakov, Alexander M.",
    title = "{Interaction of Goldstone Particles in Two-Dimensions. Applications to Ferromagnets and Massive Yang-Mills Fields}",
    doi = "10.1016/0370-2693(75)90161-6",
    journal = "Phys. Lett. B",
    volume = "59",
    pages = "79--81",
    year = "1975"
}

@inproceedings{Topping2006LecturesOT,
  title={Lectures on the Ricci Flow},
  author={Peter M. Topping},
  year={2006}
}

@article{article,
author = {Mantegazza, Carlo and Catino, Giovanni and Cremaschi, Laura and Djadli, Zindine and Mazzieri, Lorenzo},
year = {2015},
month = {07},
pages = {},
title = {The Ricci-Bourguignon Flow},
volume = {287},
journal = {Pacific Journal of Mathematics},
doi = {10.2140/pjm.2017.287.337}
}

@book{chow2004ricci,
  title={The Ricci Flow: An Introduction: An Introduction},
  author={Chow, B. and Knopf, D. and American Mathematical Society},
  isbn={9780821835159},
  lccn={04046148},
  series={Mathematical surveys and monographs},
  url={https://books.google.de/books?id=iUbzBwAAQBAJ},
  year={2004},
  publisher={American Mathematical Society}
}

@article{Perelman:2006un,
    author = "Perelman, Grisha",
    title = "{The Entropy formula for the Ricci flow and its geometric applications}",
    eprint = "math/0211159",
    archivePrefix = "arXiv",
    month = "7",
    year = "2006"
}

@article{hamilton1982,
author = "Hamilton, Richard S.",
doi = "10.4310/jdg/1214436922",
fjournal = "Journal of Differential Geometry",
journal = "J. Differential Geom.",
number = "2",
pages = "255--306",
publisher = "Lehigh University",
title = "Three-manifolds with positive Ricci curvature",
url = "https://doi.org/10.4310/jdg/1214436922",
volume = "17",
year = "1982"
}

@article{Cao2006ACP,
  title={A Complete Proof of the Poincar{\'e} and Geometrization Conjectures - application of the Hamilton-Perelman theory of the Ricci flow},
  author={Huai-dong Cao and Xiping Zhu},
  journal={Asian Journal of Mathematics},
  year={2006},
  volume={10},
  pages={165-492}
}

@book{morgan2007ricci,
  title={Ricci flow and the Poincar{\'e} conjecture},
  author={Morgan, John W and Tian, Gang},
  volume={3},
  year={2007},
  publisher={American Mathematical Soc.}
}

@article{kleiner2008notes,
  title={Notes on Perelman’s papers},
  author={Kleiner, Bruce and Lott, John},
  journal={Geometry \& Topology},
  volume={12},
  number={5},
  pages={2587--2855},
  year={2008},
  publisher={Mathematical Sciences Publishers}
}

@article{lee2022emergent,
  title={Emergent strings from infinite distance limits},
  author={Lee, Seung-Joo and Lerche, Wolfgang and Weigand, Timo},
  journal={Journal of High Energy Physics},
  volume={2022},
  number={2},
  pages={1--105},
  year={2022},
  publisher={Springer}
}

@article{banks2011symmetries,
  title={Symmetries and strings in field theory and gravity},
  author={Banks, Tom and Seiberg, Nathan},
  journal={Physical Review D},
  volume={83},
  number={8},
  pages={084019},
  year={2011},
  publisher={APS}
}

@article{banks1988constraints,
  title={Constraints on string vacua with spacetime supersymmetry},
  author={Banks, Thomas and Dixon, Lance J},
  journal={Nuclear Physics B},
  volume={307},
  number={1},
  pages={93--108},
  year={1988},
  publisher={Elsevier}
}

@article{harlow2021symmetries,
  title={Symmetries in quantum field theory and quantum gravity},
  author={Harlow, Daniel and Ooguri, Hirosi},
  journal={Communications in Mathematical Physics},
  volume={383},
  number={3},
  pages={1669--1804},
  year={2021},
  publisher={Springer}
}

@article{harlow2018constraints,
  title={Constraints on symmetry from holography},
  author={Harlow, Daniel and Ooguri, Hirosi},
  journal={arXiv preprint arXiv:1810.05337},
  year={2018}
}

@article{ooguri2020cobordism,
  title={Cobordism conjecture in AdS},
  author={Ooguri, Hirosi and Takayanagi, Tadashi},
  journal={arXiv preprint arXiv:2006.13953},
  year={2020}
}

@article{montero2021cobordism,
  title={Cobordism conjecture, anomalies, and the String Lamppost Principle},
  author={Montero, Miguel and Vafa, Cumrun},
  journal={Journal of High Energy Physics},
  volume={2021},
  number={1},
  pages={1--47},
  year={2021},
  publisher={Springer}
}

@article{dierigl2021swampland,
  title={Swampland cobordism conjecture and non-Abelian duality groups},
  author={Dierigl, Markus and Heckman, Jonathan J},
  journal={Physical Review D},
  volume={103},
  number={6},
  pages={066006},
  year={2021},
  publisher={APS}
}

@article{buratti2021dynamical,
  title={Dynamical Cobordism and Swampland Distance Conjectures},
  author={Buratti, Ginevra and Calder{\'o}n-Infante, Jos{\'e} and Delgado, Matilda and Uranga, Angel M},
  journal={Journal of High Energy Physics},
  volume={2021},
  number={10},
  pages={1--28},
  year={2021},
  publisher={Springer}
}

@article{andriot2022looking,
  title={Looking for structure in the cobordism conjecture},
  author={Andriot, David and Carqueville, Nils and Cribiori, Niccol{\`o}},
  journal={arXiv preprint arXiv:2204.00021},
  year={2022}
}

@article{angius2022end,
  title={At the End of the World: Local Dynamical Cobordism},
  author={Angius, Roberta and Calder{\'o}n-Infante, Jos{\'e} and Delgado, Matilda and Huertas, Jes{\'u}s and Uranga, Angel M},
  journal={arXiv preprint arXiv:2203.11240},
  year={2022}
}

@article{blumenhagen2022dimensional,
  title={Dimensional Reduction of Cobordism and K-theory},
  author={Blumenhagen, Ralph and Cribiori, Niccol{\`o} and Kneissl, Christian and Makridou, Andriana},
  journal={arXiv preprint arXiv:2208.01656},
  year={2022}
}

@misc{tasifth,
  doi = {10.48550/ARXIV.1806.01854},
  
  url = {https://arxiv.org/abs/1806.01854},
  
  author = {Weigand, Timo},
  
  title = {TASI Lectures on F-theory},
  
  publisher = {arXiv},
  
  year = {2018},
  
  copyright = {arXiv.org perpetual, non-exclusive license}
}

@article{kodaira,
 author = {K. Kodaira},
 journal = {Annals of Mathematics},
 number = {3},
 pages = {563--626},
 publisher = {Annals of Mathematics},
 title = {On Compact Analytic Surfaces: II},
 volume = {77},
 year = {1963}
}

@article{Neron1964,
author = {Néron, André},
journal = {Publications Mathématiques de l'IHÉS},
language = {fre},
pages = {5-128},
publisher = {Institut des Hautes Études Scientifiques},
title = {Modèles minimaux des variétés abéliennes sur les corps locaux et globaux},
volume = {21},
year = {1964},
}

@book{miranda1989basic,
  title={The Basic Theory of Elliptic Surfaces: Notes of Lectures},
  author={Miranda, R. and Universit{\`a} di Pisa. Dipartimento di matematica},
  lccn={92155770},
  series={Dottorato di ricerca in matematica / Universit{\`a} di Pisa, Dipartimento di Matematica},
  url={https://books.google.de/books?id=K0DvAAAAMAAJ},
  year={1989},
  publisher={ETS Editrice}
}

@misc{EllipticSchuettShioda,
  doi = {10.48550/ARXIV.0907.0298},
  
  url = {https://arxiv.org/abs/0907.0298},
  
  author = {Schuett, Matthias and Shioda, Tetsuji},
  
  title = {Elliptic Surfaces},
  
  publisher = {arXiv},
  
  year = {2009},
  
  copyright = {arXiv.org perpetual, non-exclusive license}
}

@article{NikulinPolarization,
  title={Factor groups of groups of automorphisms of hyperbolic forms with respect to subgroups generated by 2-reflections. Algebrogeometric applications},
  author={Viacheslav V. Nikulin},
  journal={Journal of Soviet Mathematics},
  year={1983},
  volume={22},
  pages={1401-1475}
}

@misc{JacobianK3,
  doi = {10.48550/ARXIV.2109.01929},
  
  url = {https://arxiv.org/abs/2109.01929},
  
  author = {Clingher, Adrian and Malmendier, Andreas},
  
  title = {On Neron-Severi lattices of Jacobian elliptic K3 surfaces},
  
  publisher = {arXiv},
  
  year = {2021},
  
  copyright = {arXiv.org perpetual, non-exclusive license}
}

@misc{DanFreed,
title= {Bordism: old and new},
author= {Daniel S. Freed},
url={https://web.ma.utexas.edu/users/dafr/bordism.pdf},
year={2012}
}

@article{GibbonsHawking,
title = {Gravitational multi-instantons},
journal = {Physics Letters B},
volume = {78},
number = {4},
pages = {430-432},
year = {1978},
issn = {0370-2693},
doi = {https://doi.org/10.1016/0370-2693(78)90478-1},
author = {G.W. Gibbons and S.W. Hawking},
}

@article{blumenhagen2022dynamical,
  title={Dynamical Cobordism of a Domain Wall and its Companion Defect 7-brane},
  author={Blumenhagen, Ralph and Cribiori, Niccol{\`o} and Kneissl, Christian and Makridou, Andriana},
  journal={arXiv preprint arXiv:2205.09782},
  year={2022}
}

@article{structcob,
  doi = {10.48550/ARXIV.2204.00021},
  
  url = {https://arxiv.org/abs/2204.00021},
  
  author = {Andriot, David and Carqueville, Nils and Cribiori, Niccolò},
  
  title = {Looking for structure in the cobordism conjecture},
  
  publisher = {arXiv},
  
  year = {2022},
  
  copyright = {Creative Commons Attribution Non Commercial Share Alike 4.0 International}
}

\end{document}